\documentclass[journal,draftcls,onecolumn,10pt]{IEEEtran}

\makeatletter
\let\IEEEproof\proof
\let\IEEEendproof\endproof
\let\proof\@undefined
\let\endproof\@undefined
\newcommand*{\etc}{
    \@ifnextchar{.}
        {etc}
        {etc.\@\xspace}
}
\makeatother

\usepackage{report}

\let\proof\IEEEproof
\let\endproof\IEEEendproof

\newcommand{\atilde}[1]{\tilde{#1}_{\alpha}\sbs}
\newcommand{\atildes}[2]{\tilde{#1}_{#2}\sbs}
\newcommand{\twotilde}[1]{\tilde{#1}_{2}\sbs}
\newcommand{\onetilde}[1]{\tilde{#1}_{1}\sbs}

\def\aFI{$\alpha$-Fisher information}
\def\apower{$\alpha$-power }
\def\aFI{$\alpha$-Fisher information }

\usepackage{cite}

\begin{document}

\title{\Large\bfseries Information Measures, Inequalities and
  Performance Bounds for Parameter Estimation in Impulsive Noise
  Environments \thanks{This work was supported by AUB's University
    Research Board, and the Lebanese National Council for Scientific
    Research (CNRS-L). Partial results were presented in part at the
    2016 IEEE International Symposium on Information Theory.}}
 
\author{
  \authorblockN{Jihad Fahs, Ibrahim Abou-Faycal} \\
  \authorblockA{Dept.\ of Elec.\ and Comp.\ Eng.,
    American University of Beirut \\
    Beirut 1107 2020, Lebanon \\
    {\tt \{jjf03, Ibrahim.Abou-Faycal\}@aub.edu.lb}}
}

\maketitle


\begin{abstract}
  Recent studies found that many channels are affected by additive
  noise that is impulsive in nature and is best explained by
  heavy-tailed symmetric alpha-stable distributions. Dealing with
  impulsive noise environments comes with an added complexity with
  respect to the standard Gaussian environment: the alpha-stable
  probability density functions do not possess closed-form expressions
  except in few special cases. Furthermore, they have an infinite
  second moment and the ``nice'' Hilbert space structure of the space
  of random variables having a finite second moment --which represents
  the universe in which the Gaussian theory is applicable, is lost
  along with its tools and methodologies.

  This is indeed the case in estimation theory where classical tools
  to quantify the performance of an estimator are tightly related
  to the assumption of having finite variance variables. In
  alpha-stable environments, expressions such as the mean square error
  and the Cramer-Rao bound are hence problematic.
  
  In this work, we tackle the parameter estimation problem in
  impulsive noise environments and develop novel tools that are
  tailored to the alpha-stable and heavy-tailed noise environments,
  tools that coincide with the standard ones adopted in the Gaussian
  setup; namely a generalized ``power'' measure and a generalized
  Fisher information. We generalize known information inequalities
  commonly used in the Gaussian context: the de Bruijn's identity, the
  data processing inequality, the Fisher information inequality, the
  isoperimetric inequality for entropies and the Cramer-Rao
  bound. Additionally, we derive upper bounds on the differential
  entropy of independent sums having a stable component.
  Intermediately, the new power measure is used to shed some light on
  the additive alpha-stable noise channel capacity in a setup that
  generalizes the linear average power constrained AWGN channel.  Our
  theoretical findings are paralleled with numerical evaluations of
  various quantities and bounds using developed {\em Matlab} packages.

  \begin{flushleft}
  \end{flushleft}
\end{abstract}


\section{Introduction}
\label{sc:Introduction}

The presence of impulsive-noise such as those with alpha-stable
statistics, is rather frequent in communications theory. Indeed,
interference has been often found to be of impulsive nature and is
best explained by alpha-stable distributions. This is the case
for telephone noise \cite{stuck} and audio noise signals
\cite{geo}. Furthermore, in many works that treated the multiuser
interference in radio communication networks, a theoretical
derivation, based on the assumption that the interferers are
distributed over the entire plane and behave statistically as a Point
Poisson Process (PPP), yielded an interference with alpha-stable
statistics, starting with Sousa \cite{sousa92} who computed the self
interference, considering only the pathloss effect for three spread
spectrum schemes, direct sequence with binary phase shift keying
(DS/BPSK), frequency hopping with M-ary frequency shift keying
(FH/MFSK), and frequency hopping with on-off keying (FH/OOK), where a
sinusoidal tone is transmitted as the ``on" symbol. In \cite{jacek},
the authors introduced a novel approach to stable noise modeling based
on the LePage series representation which permits the extension of the
results on multiple access communications to environments with
lognormal shadowing and Rayleigh fading. Continuous time multiuser
interference was also found~\cite{hughes} to be well represented as an
impulsive alpha-stable random process. Recently in~\cite{Win},
alpha-stable distributions were found to model well the aggregate
interference in wireless networks: the authors treated the problem in
a general framework that accounts for all the essential physical
parameters that affect network interference with applications in
cognitive radio, wireless packets, covert military schemes and
networks where narrowband and ultra-wide band systems
coexist. In~\cite{evans}, Gulati et al. showed that the
statistical-physical modeling of co-channel interference in a field of
Poisson and Poisson-Poisson clustered interferers obeys an
alpha-stable or Middleton class A statistics depending whether the
interferers are spread in the entire plane, in a finite area or in a
finite area with a guard zone with the alpha-stable being suitable for
wireless sensor, ad-hoc and femtocells networks when both in-cell and
out-of-cell interference are included. A generalization of the
previous results for radio frequency interference in multiple antennas
is found in \cite{chopra} where joint statistical-physical
interference from uncoordinated interfering sources is derived without
any assumption on spatial independence or spatial isotropic
interference. Lastly, the alpha-stable model arises as a suitable
noise model in molecular communications~\cite{Farsa15}.

An important problem in the theory of non-random parameter estimation
is to find ``good'' estimators of some quantity of interest based on a
given observation. Generally, this is done by using a quality measure
of the estimator's (average) performance: the Mean Square Error
(MSE). The use of the MSE is tightly related to the assumption of
finite variance noise and one can even argue that it is related to a
``potential Gaussian'' setup. Naturally, under this finite-variance
assumption, one can restrict the quest of finding ``good'' estimators
to the Hilbert space of finite second moment Random Variables (RV)s
which leads to the well-established ``Gaussian'' or ``linear''
estimation theory. When the observation is contaminated with an
impulsive noise perturbation --having an infinite variance,
restricting the look-up universe for good estimators to that of finite
variance RVs is no longer optimal neither necessarily
sensible. Additionally, tools such as the MSE will turn out to be
problematic.
  
In this work we consider the non-random parameter estimation problem
whereby we want to estimate a non-random parameter(s) $\theta \in
\Reals \,(\Reals^{d})$ based on a noisy observation $X = \theta + N$
and where the additive noise $N$ is of impulsive nature. In the case
where the noise $N$ has a finite variance, the problem is
well-understood: let $\hat{\theta}(X)$ be an estimator of
$\theta$ based on observing $X$, then
\begin{itemize}
\item The quality of the estimator is measured via the MSE: ``$\,\E{
    \left| \hat{\theta}(X) - \theta \right|^2 }$''. Hence, Minimum
  Mean Square Error (MMSE) estimators are optimal.
\item A lower-bound on the MSE of the estimator is given by the
  Cramer-Rao (CR) bound:
  \begin{equation}
    \label{eq:CRS}
    \E{ \left| \hat{\theta}(X) - \theta \right|^2 } \geq \frac{1}{J(N)},
  \end{equation}
  where $J(N)$ is the Fisher information\footnote{The Fisher
    information $J(Y)$ of a RV $Y$ having a Probability Density
    Function (PDF) $p(y)$ is defined as:
    \begin{equation*}
      J(Y) = \int_{-\infty}^{+\infty}\frac{1}{p(y)}\,p'^{2}(y) \,dy,
    \end{equation*}
    whenever the derivative and the integral exit.} of the RV $N$.
\item Equality holds in equation~(\ref{eq:CRS}) whenever $N$ is
  Gaussian distributed and $\hat{\theta}(X) = X$ is the Maximum
  Likelihood (ML) estimator.
\end{itemize}
  
In order to understand the parameter estimation problem in the
impulsive noise scenario, one must answer the following:
\begin{itemize}
\item[1-] Under the impulsive noise assumption, the MMSE estimator is
  not necessarily optimal and the linear MSE estimator is not
  sensible. ``Good'' estimators candidates might possibly have an
  infinite second moment which implies that a new quality measure has
  to be defined. This quality measure is to be interpreted as the
  average ``strength'' or power of the estimation error.
\item[2-] Since the Fisher information $J(\cdot)$ is tightly related
  to Gaussian variables through de Bruijn's identity\footnote{The de
    Bruijn's identity is defined as: For any $\epsilon \geq 0$,
    \begin{equation}
      \label{debruijn}
      \frac{d}{d \epsilon}\,h(X + \sqrt{\epsilon} Z) = \frac{\sigma^2}{2} J(X + \sqrt{\epsilon} Z),
    \end{equation}
    where $Z$ is independent of $X$, Gaussian distributed with mean
    $0$ and variance $\sigma^2$}, a new information measure has to be
  defined-- one that is adapted to impulsive noise variables.
  Similarly to $J(\cdot)$, the new information measure is to be
  related to the alpha-stable distribution through a de Bruijn's type
  of equality.
\item[3-] Establishing a new CR bound: the new quality measure of an
  estimator is to be lower bounded, function of the inverse of the new
  information measure.
\end{itemize}
When it comes to objective 1, a survey of the literature shows that
few alternative measures of power were proposed:

\begin{itemize}
\item In~\cite{Shao}, Shao and Nikias proposed the ``dispersion'' of a
  RV as a measure that plays a similar role to the variance. However,
  since no analytical expression is defined for the dispersion except
  for alpha-stable distributions, they propose the usage of the
  Fractional Lower Order Moments (FLOM) $\E {|X|^r}$ ($r<2$) as an
  alternative which yields a non-linear signal processing theory.


\item Based on logarithmic moments of the form $\E{\log|X|} $, an
  alternative notion of power was introduced by Gonzales~\cite{gon}
  for heavy-tailed distributions which he labeled as the Geometric
  Power (GP):
  \begin{equation*}
    \mathcal{S}_0(X) \eqdef e^{\E{\log|X|}}.
  \end{equation*} 
  
  The author considered logarithmic moments as a ``universal
  framework'' for dealing with algebraic tail processes that will
  overcome the shortcomings of the FLOM approach which he summarized
  by the fact that no appropriate value of $r$ is feasible for all
  impulsive processes.
  Also the discontinuity in the FLOM is yet another unpleasant
  feature. In fact, for a given $0<r<2$, two alpha stable
  distributions with characteristic exponents $\alpha_1 = r +
  \epsilon$ and $\alpha_2 = r - \epsilon$ (for some $\epsilon>0$),
  will respectively have a finite and infinite $r$-th absolute moment,
  though one can agree that they would have similar statistical
  behavior. However, all stable distributions have a finite
  logarithmic moment~\cite{gon}.
\end{itemize}

The GP was used in formulating new impulsive signal processing
techniques with the proposition of new types of non-linear filters
referred to as ``myriad filters'', which are basically Maximum
Likelihood (ML) estimators of the location of a Cauchy distribution
with an optimality tune
parameter~\cite{gonza11}.  
However, the GP suffers from a serious drawback since for any variable
$X$ that has a mass point at zero, $\mathcal{S}_0(X)$ will be
necessarily null even if say other non-zero mass points are
existent. This would yield a zero power for a non-zero signal.

In relation to objective 2, generalizing ``Gaussian''
information-theoretic properties and tools to ``stable'' ones is done
in~\cite{john2013} where a new score function is defined in terms of a
scaled conditional expectation and a de Bruijn's identity is found in
terms of the new score function in a relative manner with respect to
that of a stable variable.  Recently in~\cite{toscani15}, Toscani
proposed a fractional score function using fractional derivatives and
defined a fractional Fisher information that evaluates to infinity for
stable variables.  Using it in a relative manner --with respect to
stable variables, the relative fractional Fisher information is found
to satisfy a Fisher information inequality and is used to find the
rate of convergence in relative entropy of scaled Independent and
Identically Distributed (IID) sums to stable variables.

Up to the authors' knowledge, objective 3 has only been addressed
  in~\cite{lutwak2005} where the authors derived a Cramer-Rao type of
  inequality featuring the finite fractional moment of order $r \geq
  1$ of a variable and a generalized Fisher information. The work
  in~\cite{lutwak2005} was in the direction of extending information
  theoretic inequalities to new ones where generalized Gaussians are
  extremal distributions rather than characterizing the quality of
  estimators in impulsive noise environments. We also note that the CR
  result in~\cite{lutwak2005} suffers from the restriction of having
  variables with finite fractional moments of order $r \geq 1$ which
  is not the case in this paper where variables with only finite
  logarithmic moments are considered.

Naturally, this parameter estimation problem is also that of
estimating the location parameter of an alpha-stable
variable. Previous works that treated the estimation of the various
parameters of alpha-stable distribution~\cite{fama1971, Ruth1980,
  mccull1986, ma1995, tsih1996, Ercan, nolan2001} had a primary goal
of finding specific estimators. They are based on heuristics for which
the authors either conducted consistency or asymptotic analysis, or
tested empirical evaluations versus numerical computations and
Monte-Carlo simulations in order to validate and evaluate the proposed
estimators. In this context, in this work we define and find quality
measures and universal bounds that are satisfied by all location
parameter estimators of impulsive distributions.  Our main
contributions are four fold:
\begin{enumerate}
\item \underline{A generalized power notion}: The evaluation of
  performance measures in multiple applications in communications
  theory is generally done function of the channel state quality. A
  key quantity that summarizes the quality of the channel is the
  Signal-to-Noise Ratio (SNR) which is a ratio between the power of a
  signal containing the relevant information to that of the noise
  signal. A standard measure of the signal power is made through the
  evaluation of the second moment. When working in alpha-stable noise
  environments, some information bearing signals will necessarily have
  an infinite second moment which eventually leads to having zero
  SNRs, a fact that masks the possibility to quantify the channel's
  state. We propose in Definition~\ref{powdef} a new ``relative''
  power measure that we call the {\em \apower\sbs}: a strength
  measure that takes into account the type of the disturbing
  noise. This would seem reasonable whenever the goal of the
  communication system is to maintain a Quality of Service (QoS) level
  for some or all of its users which is translated, for example, to a
  threshold rate (output entropy) or an output SNR. In both cases the
  QoS will be dependent on the output signal. Our ``output''-based
  approach is tailored to this type of applications since it focuses
  on the output signal and takes into account the type of the
  encountered noise in the received signal in order to define sensible
  tools to quantify the QoS criteria. As an example, we derive in
  Theorem~\ref{th:capstablechannel} the capacity of an additive stable
  noise channel under a constraint on its output's \apower.

  Another application is the parameter estimation problem where the
  observed output, affected by stable noise is sufficient for the
  characterization of the estimator's performance. The generalized
  power measure is chosen in such a way that when constraining it,
  stable variables will be entropy maximizers, proven in
    Theorem~\ref{th:maxentro}. It is then shown to comply with
  generic properties that are satisfied by the standard deviation and
  is numerically evaluated for different types of probability
  densities.

\item \underline{A generalized information measure}: We consider an
  alternative formulation of the Fisher information that is more
  relevant than $J(X)$ when dealing with RVs corrupted by additive
  noise of infinite second moment; In essence, our starting point is
  one where --in a similar fashion to the Gaussian case-- we enforce a
  generalized de Bruijn identity to hold: motivated by the fact that
  the derivative of the differential entropy with respect to small
  variations in the direction of a Gaussian variable is a scaled
  $J(\cdot)$, we propose in Definition~\ref{def:Jalpha} a new
  notion of Fisher information as a derivative of differential entropy
  in the direction of infinitesimal perturbations along stable
  variables and we label it the ``Fisher information of order
  $\alpha$''  or the {\em \aFI\bs.} Next, we derive in
    Lemma~\ref{lem2} an integral expression for the new quantity that
  is a generalization of the well-known expression of the Fisher
  information.
  We note that the definition of the \aFI in this manuscript is
  an absolute measure and different from the one
  in~\cite{toscani15}. It has different usages and applications and
  was independently developed.

\item \underline{Generalized information-theoretic inequalities}:
  Information inequalities have been investigated since the foundation
  of information theory. It started with Shannon~\cite{Sha48_1} with
  the fact that Gaussian distributions maximize entropy under a second
  moment constraint. Then a lower bound on the entropy of independent
  sums of RVs, commonly known as the Entropy Power Inequality (EPI)
  was proved. The EPI states that given two real independent RVs $X$,
  $Z$ such that $h(X)$, $h(Z)$ and $h(X + Z)$ exist, then~(Corollary
  3, \cite{bob15})
  \begin{equation}
    N(X + Z) \geq N(X) + N(Z),
    \label{EPI}
  \end{equation}
  where $N\left(X\right)$ is the {\em entropy power\/} of $X$ and is
  equal to
  \begin{equation*}
    N\left(X\right) = \frac{1}{2 \pi e} e^{2 h(X)}.
  \end{equation*}
  The EPI was a proposition of Shannon who provided a local proof.
  Later Stam~\cite{sta} followed by Blachman~\cite{bla} presented
  complete proofs.  These proofs of the EPI relied on two information
  identities: the de Bruijn's identity and the Fisher Information
  Inequality (FII). The latter states that given $X$ and $Z$ two
  independent RVs such that the respective {\em Fisher information}
  $J(X)$ and $J(Z)$ exist. Then
  \begin{equation}
    \frac{1}{J(X+Z)} \geq \frac{1}{J(X)} + \frac{1}{J(Z)}.
    \label{FII}
  \end{equation}
  
  The remarkable similarity between~ equations (\ref{EPI})
  and~(\ref{FII}) was pointed out in Stam's paper~\cite{sta} who in
  addition, related the entropy power and the Fisher information by an
  ``uncertainty principle-type'' relation:
  \begin{equation}
    \label{pro}
    N(X)J(X) \geq 1,
  \end{equation}
  which is commonly known as the Isoperimetric Inequality for
  Entropies (IIE)~\cite[Theorem 16]{cov-dem}.  Interestingly, equality
  holds in~equation (\ref{pro}) whenever $X$ is Gaussian distributed
  and in equations~(\ref{EPI})--(\ref{FII}) whenever $X$ and $Z$ are
  independent Gaussians. As it can be noticed, the previously cited
  inequalities revolve around Gaussian variables. When it comes to the
  general stable family, the relative fractional Fisher information
  defined in~\cite{toscani15} is found to satisfy a Fisher information
  inequality and is used to find the rate of convergence in relative
  entropy of scaled IID sums to stable variables.  In this paper, we
  generalize these information theoretic inequalities that are based
  on the Gaussian setting to generic ones in the stable setting which
  coincide with the regular identities in the Gaussian setup. Namely,
  when restricted to the range $1 < \alpha \leq 2$, the \aFI is
    found in Theorem~\ref{GFIItheorem} to satisfy a Generalized
  Fisher Information Inequality (GFII). Then, we use the GFII and
  the generalized de Bruijn (proven in
    Theorem~\ref{thm:debruijn}) to derive in Theorem~\ref{bdsta}
  an upper bound on the differential entropy of the independent sum of
  two RVs where one of them is stable. Finally, in
    Theorem~\ref{thGII} a Generalized Isoperimetric Inequality for
  Entropies (GIIE) is proved to hold.

\item \underline{A Generalized Cramer-Rao bound}: Well-known
  identities such as the Cramer-Rao bound which provides a lower bound
  on the mean square error of estimators in the from of the inverse of
  $J(X)$ are adequate in the finite variance setup. If the observed
  noisy variable has an infinite second moment, the use of the
  Cramer-Rao bound in its classical form is problematic. We derive in
  Theorem~\ref{th:cramer-rao} a generalized Cramer-Rao bound,
  that relates the ``relative'' power of the estimation error to the
  generalized Fisher information $J_{\alpha}(\cdot)$.
\end{enumerate}

The rest of this paper is organized as follows: we propose in
Section~\ref{sec:RPM} the \apower, a generalized power parameter
and we provide some of its properties and applications. We define in
Section~\ref{sc:GFI} the \aFI, we list its properties and we
establish a generalized de Bruijn's identity. In
Section~\ref{sec:GITI}, information inequalities are shown to be
satisfied by the the \aFI with applications in finding upper
bounds on the differential entropy of independent sums when one of the
variables is stable and establishing a generalized IIE. The
generalized CR bound is stated and proved in Section~\ref{sec:CRbd}
and Section~\ref{con:Conclusion} concludes the paper.


\section{The \apower\sbs, a Relative Power Measure}
\label{sec:RPM}

Power measures are important tools that can provide partial yet
fundamental information about a signal. They serve multiple purposes
such as signal strength comparisons or as reference units for the
computation of performance and quality indicators.  We stipulate that
a ``strength'' or power measure $\Pw{\vX}$ of a random vector $\vX$
should satisfy the following:
\begin{itemize}
\item[R1-] $\Pw{\vX} \geq 0$, with equality if and only if ${\bf
    X}=0$.
\item[R2-] $\Pw{a \vX} = |a| \Pw{\vX}$, for all $a \in \Reals$.
\end{itemize}

These restrictions are ``minimal'' and do not contain for example some
of the dispensable properties satisfied by the GP such as the
multiplicativity and the triangular inequality
properties~\cite{gon}. However they are deemed sufficient to define a
strength measure.

As a notion of average power, the second moment is the answer to a
widely known result in communications theory; it is the constraint
under which a Gaussian density function is entropy maximizer. In order
to come up with a notion of average power in the presence of
alpha-stable distributions, one might consider adopting the
measure/constraint under which sub-Gaussian~\footnote[3]{In some
  texts, the term sub-Gaussian refers to distribution functions whose
  tails are faster than those of a Gaussian. In this work, we do not
  use the term sub-Gaussian in this sense.} symmetric
alpha-stable~\footnote[4]{We use the term symmetric alpha-stable
  (\SaS) to refer to the class of non-degenerate symmetric stable
  distributions excluding the Gaussian. Otherwise, only the term
  symmetric stable (SS) will be used.} (\SaS) density functions with
an underlying Gaussian vector having IID zero-mean components (refer
to Definition~\ref{def:subgauss}, Appendix~\ref{asd}) are entropy
maximizers; an approach that we adopt in what follows.


\subsection{A Power Parameter in the Presence of Stable
  Variables}
\label{seclocpow}

In this manuscript we denote by $\atilde{\vZ} \sim
\bm{\mathcal{S}} \left(\alpha, \left( \frac{1}{\alpha}
  \right)^{\frac{1}{\alpha}} \right)$ a reference {\em Symmetric
  Stable (SS) vector\/}, i.e.,
\begin{itemize}
\item[$\bullet$] \underline{whenever $\alpha \neq 2$:} a reference
  sub-Gaussian S$\alpha$S vector with an underlying Gaussian vector
  having IID zero-mean components with variance $\sigma^2 =
  2\left(\frac{1}{\alpha}\right)^{\frac{2}{\alpha}}$.
\item[$\bullet$] \underline{when $\alpha =2$:} a reference
  Gaussian vector of IID components with mean zero and variance 1.
\end{itemize}

\begin{definition}[Power Parameter]
  \label{powdef}
  The ``power of order $\alpha$'' or {\em \apower\/} of
  non-zero random vector $\vX$ is the non-negative scalar
  $\aPw{\vX}$ such that:
  \begin{equation}
    \label{newpow}
    - \E{ \ln p_{{\bf \atilde{Z}}} \left( \frac{{\vX}}{ \aPw{\vX} } \right) } = h(\atilde{\vZ}),
  \end{equation}  
  where $h(\atilde{\vZ})$ is the differential entropy of
  $\atilde{\vZ}\,$. For the deterministic $\vX = {\bf 0}$, we
  define $\aPw{\vX} = 0$.
\end{definition}
     
The existence and uniqueness of the \apower will be addressed
shortly. Intuitively, one may think of $\aPw{\vX}$ as a
``relative power'' with respect to $\atilde{\vZ}$ which is a
reference variable whose \apower is equal to unity.  In the two
special cases where closed-form expression of the PDF is available,
the \apower can be evaluated:
\begin{itemize}
\item When $\alpha = 2$, $\twotilde{\vZ}$ is a zero-mean Gaussian
  vector with identity covariance matrix and
  \begin{equation*}
    \twoPw{\vX}= \sqrt{\frac{ \E{ \| \vX \|^2 } }{d}}.
  \end{equation*}
\item When $\alpha =1$ (see~\cite{sam1994}),
  \begin{equation*}
    p_{{\bf \onetilde{Z}}}(\vx) = \frac{1}{\pi^{\frac{d+1}{2}}}\Gamma\left(\frac{d+1}{2}\right) 
    \, \frac{1}{\left(1+\|\vx\|^2\right)^{\frac{d+1}{2}}},
  \end{equation*}
  and $\onePw{\vX}$ is the solution of 
  \begin{equation*}
    \E { \ln \left( 1+ \left\| \frac{\vX}{\onePw{\vX}} \right\|^2 \right) }
    = \E{ \ln \left( 1+ \left\| \onetilde{\vZ} \right\|^2 \right)}.
  \end{equation*}
\end{itemize}

 
As defined in~(\ref{newpow}), the quantity $\aPw{\vX}$ is endowed
  with some power properties that we list hereafter and prove in
  Appendix~\ref{ap:prop_aPw}.


\begin{property}
  \label{prop:aPw}
  Let $\vX$ and $\vY$ be random vectors such that:
    \begin{equation*} \left\{ \begin{array}{ll}
          \E{\ln\left(1 + \| \vX \|\right)} < \infty, & \text{ when considering cases where } \alpha < 2. \\
          \E{ \| \vX \|^2} < \infty & \text{ when considering cases where } \alpha = 2.
        \end{array} \right.
    \end{equation*} 
    The following properties hold:

  \begin{itemize}
  \item[(i)] The \apower $\aPw{\vX}$ exists, is unique and
      $\aPw{\vX} \geq 0$ with equality if and only if $\vX = {\bf
        0}$.
    
  \item[(ii)] For any $a \in \Reals, \aPw{a \vX} = |a| \, \aPw{\vX}$.
   
  \item[(iii)] If $\vX$ and $\vY$ are independent and $\vY$ has a
      rotationally symmetric PDF that is non-increasing in $\| \vy
      \|$,
      then $\aPw{\vX + \vY} \geq \aPw{\vY}$.

  \item[(iv)] If $\vX$ and $\vY$ are independent and $\vY$ has a
      rotationally symmetric PDF that is non-increasing in $\| \vy
      \|$, then $\aPw{c \vX + \vY}$ is non-decreasing in $|c|$, $c \in
      \Reals$.

  \item[(v)] Whenever $\vX \sim \bm{\mathcal{S}} \left( \alpha,
        \gamma_{\vX} \right)$, $\aPw{\vX} =
      (\alpha)^{\frac{1}{\alpha}} \gamma_{\vX}$.
  \end{itemize}
\end{property}

Though the definition of the \apower as stated in
  equation~(\ref{newpow}) is implicit and dependent on the density
  function of the SS vector ${\bf \atilde{Z}}$ which does not have
  closed form expression except in the special cases of the Cauchy and
  the Gaussian distributions, the computation of the \apower
  $\aPw{\vX}$ of a certain random vector $\vX$ can be done efficiently
  using numerical computations. In fact, the stable densities can be
  computed numerically as inverse Fourier transforms or by using {\em
    Matlab} packages that compute these densities such as the ``{\em
    Stable}'' package provided by Nolan~\cite{nolan-software}. We use
  here the latter and we develop a {\em Matlab} code that computes the
  \apower for a scalar RV according to Definition~\ref{powdef}. We
  plot in Figure~\ref{fig:figurepow}, the \apower of several
  probability laws-- Gaussian, uniform, Laplace, Cauchy and
  alpha-stable
  , with respect to a multitude of symmetric alpha-stable
  distributions with the characteristic exponent $\alpha$ ranging from
  $0.4$ to $1.8$.
\begin{itemize}
\item[-] Consider for example $\atildes{Z}{1.2\,}$. The \apower of a
  Gaussian variable $X \sim \mathcal{N}(0,2)$ is equal to
  $\aPws{X}{1.2} = 0.7869$. Using the scalability property (ii), the
  \apower of a Gaussian variable $X \sim \mathcal{N}(0,\sigma^2)$ is
  equal to $\aPws{X}{1.2} = 0.7869 \frac{\sigma}{\sqrt{2}} = 0.5564 \,
  \sigma$. Note that as already known, the power the Gaussian variable
  $X \sim \mathcal{N}(0,\sigma^2)$ with respect to $\twotilde{Z} \sim
  \mathcal{N}(0,1)$ is equal to $\twoPw{X} = \sigma$. 
\item[-] Another example is when $X \sim \mathcal{U}[-a,+a]$, a
    uniform RV with zero mean and variance equal to $\frac{a^2}{3}$.
    With respect to $\atildes{Z}{0.8}$, it has an \apower of
    $\aPws{X}{\;0.8} = 0.3036 \frac{a}{\sqrt{3}} = 0.1753 \,a$,
    whereas with respect to the Gaussian law the power is equal to the
    standard deviation $\twoPw{X} = \frac{a}{\sqrt{3}} = 0.5774 \,a$.
\end{itemize}

\begin{figure}[htb]
  \begin{center}
    \includegraphics[width=5.5in]{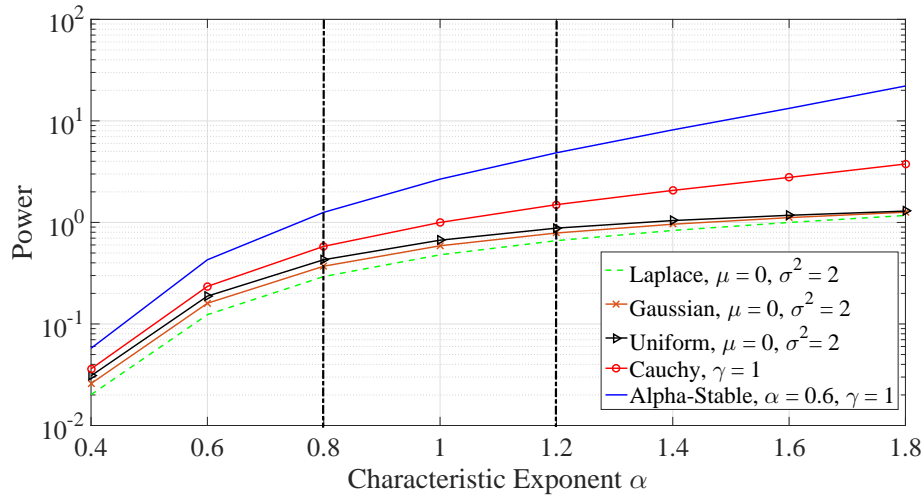}
    \caption{\small Evaluation of the \apower of some probability laws
      with respect to $\atilde{Z} \sim \mathcal{S} \left( \alpha,
        (\alpha)^{-\frac{1}{\alpha}}\right)$ for different values of
      $\alpha$.
      \label{fig:figurepow}}
  \end{center}
\end{figure}

For the value of the exponent $\alpha =2$, the \apower of the Cauchy
and alpha-stable laws evaluate to infinity and they are not shown in
Figure~\ref{fig:figurepow}.

\subsection{Applications}

The new ``power'' measure may be used in a variety of setups. We
  showcase in what follows two scenarios related to two fundamental
information-theoretic problems: entropy maximization and channel
capacity evaluation.

\subsubsection{Stable Maximizing Entropy}
\label{stamax}

Having adopted a generic power definition when considering stable
noise environments, we study the solution of the entropy maximization
problem subject to a constraint on the newly defined power. Namely,
let $\text{P} > 0$ and consider the set of random variables whose
\apower is equal to $\text{P}$:
\begin{equation*}
  \set{P} = \left\{\text{distribution functions} \,\, F \,\, \text{on}\,\, \Reals^d : -\int \ln p_{\bf \atilde{Z}}
    \left(\frac{\bf x}{\text{P}}\right)\,dF(\vx) = h({\bf \atilde{Z}})\right\}.
\end{equation*}
According to~\cite[Section 12.1]{cover}, among all distribution
functions $F \in \set{P}$, the one that maximizes differential entropy
has the following PDF:
\begin{equation*}
  p^*(\vx) = e^{\lambda_0 + \lambda_1 \ln p_{\bf \atilde{Z}}\left(\frac{\bf x}{\text{P}}\right)} = e^{\lambda_0} p_{\bf \atilde{Z}}^{\lambda_1}
  \left(\frac{\bf x}{\text{P}}\right),
\end{equation*}
where $\lambda_0$ and $\lambda_1$ are chosen so that $p^*(\vx) \in
\set{P}$. Since $\frac{1}{\text{P}} \, p_{\bf
  \atilde{Z}}\left(\frac{\bf x}{\text{P}}\right)$ is of the sought
after form,
\begin{equation}
  \label{maxh}
  \arg\max_{F \in \set{P}} \, h(F) = \text{P} {\bf \atilde{Z}} \sim \bm{\mathcal{S}}\left(\alpha,
    \left(\frac{1}{\alpha}\right)^\frac{1}{\alpha} \text{P}\right),
\end{equation}
and the value of the maximum is:
\begin{equation*}
  h(F^*) = h({\bf \atilde{Z}}) + d \ln \text{P}.
\end{equation*}

As a direct generalization, one can write:
\begin{theorem}
  \label{th:maxentro}
  Let
  \begin{equation}
    \label{setmax}
    \set{P}_A = \left\{\text{distribution functions} \,\, F \,\, \text{on}\,\, \Reals^d: 0< \aPw{F} \leq A, \, A>0 \right\}.
  \end{equation}
  Then
  \begin{equation*}
    A {\bf \atilde{Z}} = \arg\max_{F \in \set{P}_A} \, h(F) ,
  \end{equation*}
  and the maximum entropy value is $h({\bf \atilde{Z}}) + d \ln A$ .
\end{theorem}

\subsubsection{Communicating over Stable Channels}

Consider the additive linear channel:
\begin{equation}
  \label{chanstable}
  \vY = \vX + \vN,
\end{equation}
where $\vY$ is the channel output, $\vX$ is the input and $\vN \sim
\bm{\mathcal{S}} \left( \alpha, \gamma_{\bf N} \right)$ is the
additive SS noise vector which is independent of $\vX$. We ask the
following question: what constraint is to be imposed on the input such
that a stable input achieves the capacity of
channel~(\ref{chanstable})?  Under this scenario, and knowing that a
stable input generates a stable output, a sufficient condition is that
the output space induced by the channel is the space where a stable
variable maximizes entropy, specifically a space of a form as
  in~(\ref{setmax}). This leads to the following result:

\begin{theorem}
  \label{th:capstablechannel}
  Let $\vN \sim \bm{\mathcal{S}} \left( \alpha, \gamma_{\bf N}
  \right)$ and $A$ be a non-negative scalar such that $A \geq
  \aPw{\vN} = (\alpha)^{\frac{1}{\alpha}} \gamma_{\bf N}$. Consider
  the space $\mathcal{P}_{\left[\aPw{\vN}, A \right]}$ of probability
  distributions:
  \begin{equation}
    \label{outcond}
    \set{P}_{\left[\aPw{\vN}, A \right]} = \left\{\text{distribution functions} \,\, F \,\, \text{on}\,\, \Reals^d:
      \aPw{\vN} \leq \aPw{F} \leq A \right\}.
  \end{equation}
  Whenever the output $\vY$ of channel~(\ref{chanstable}) is subjected
  to~(\ref{outcond}), the channel capacity evaluates to:
  \begin{equation*}
    C = d \ln \left( \frac{A}{\aPw{\vN}} \right) = d \ln\left(\SNR_{\text{output}}\right),
  \end{equation*}
  and is achieved by $\vX^* \sim \bm{\mathcal{S}} \left( \alpha,
    \left( \frac{1}{\alpha} \right)^{\frac{1}{\alpha}}
    \sqrt[\alpha]{A^{\alpha} - \aPw{\vN}^{\alpha} }
  \right)$. Furthermore, the input cost constraint can be written as:
  \begin{equation}
    \exists \, \text{P} \in \left[ \aPw{\vN}, A \right], \qquad 
    \Ep{\vX}{D\left(p_{\aPw{\vN} {\bf \atilde{Z}}} (\vv-\vX) \big\| p_{\text{P} {\bf \atilde{Z}}}({\bf v})\right)}
    = \ln \frac{\text{P}}{ \aPw{\vN}}.
    \label{coscons}
  \end{equation}
\end{theorem}

\begin{proof}
  By Theorem~\ref{th:maxentro}, under condition~(\ref{outcond}) a
  stable output $\vY^* \sim \bm{\mathcal{S}} \left( \alpha, \left(
      \frac{1}{\alpha} \right)^{\frac{1}{\alpha}} A \right)$ maximizes
  the output entropy and achieves the channel capacity $C$:
  \begin{eqnarray*}
    C = h(\vY^*) - h(\vN) &=&d \ln(A) + h({\bf \atilde{Z}}) - d \ln( \aPw{\vN} ) - h({\bf \atilde{Z}})\\
    &=& d \ln \left( \frac{A}{\aPw{\vN}} \right),
  \end{eqnarray*}
  where we used the fact that $h(\vN) = \ln( \aPw{\vN}) + h({\bf
    \atilde{Z}})$ since $\gamma_{\bf N} = \aPw{\vN} \gamma_{\bf
    \atilde{Z}}$. The optimal input $\vX^*$ which yields $\vY^*$ is
  also distributed according to a stable variable with parameter
  $\gamma_{\vX^*}$:
  \begin{equation*}
    \gamma^{\alpha}_{\vX^*} =  \gamma^{\alpha}_{\vY^*} -  \gamma^{\alpha}_{\vN} =  \frac{1}{\alpha} \left( A^{\alpha} 
      - \aPw{\vN}^{\alpha} \right),
  \end{equation*}
  which by property (v) yields,
  \begin{equation*}
    \aPw{\vX^*}^{\alpha} = \alpha \gamma^{\alpha}_{\vX^*} = A^{\alpha} - \aPw{\vN}^{\alpha}.
  \end{equation*}
  Finally, we determine below the input cost constraint that yields
  the output space $\set{P}_{\left[ \aPw{\vN}, A \right]}$. The output
  condition~(\ref{outcond}) is explicitly stated as the space of all
  random vectors $\vY$ such that there exists a $\text{P} >0$, such
  that $\aPw{\vN} \leq \text{P} \leq A$ and
  \begin{eqnarray}
    -\E{\ln p_{\bf \atilde{Z}} \left( \frac{\vY}{P} \right) } = h({\bf \atilde{Z}}) \quad \iff \quad 
    \Ep{\vX}{ -\Ep{\vN}{ \ln p_{\bf \atilde{Z}} \left( \frac{\vX + \vN}{P}\right) \Bigg| \vX }}
    = h({\bf \atilde{Z}}) \label{costdef0},
  \end{eqnarray}
  where we used the iterated expectations to write the second
  equation. Equation~(\ref{costdef0}) can be interpreted as the input
  cost function $\mathcal{C}(\cdot)$ being
  \begin{equation}
    \label{cosfun}
    \mathcal{C} \left( \vx, \text{P}\right) = -\Ep{\vN}{ \ln p_{\bf \atilde{Z}} \left(\frac{\vx + \vN}{P} \right) },
  \end{equation} 
  and the input cost constraint being:
  \begin{equation*}
    \exists \, \text{P} \in \left[ \aPw{\vN}, A \right], \qquad \Ep{\vX}{ \mathcal{C} \left( \vX, \text{P} \right) } = h({\bf \atilde{Z}}).
  \end{equation*}

  The cost function and the cost constraint can be written in a
  different form:
  \begin{eqnarray}
    \mathcal{C}(\vx,\text{P}) &=& - \int p_{\aPw{\vN} {\bf \atilde{Z}}} (\vn) \, \ln p_{\bf \atilde{Z}}
    \left(\frac{\vx+\vn}{P}\right)\, d \vn \nonumber\\
    &=& - \int p_{\aPw{\vN} {\bf \atilde{Z}}}\left(\vv-\vx\right) \, \ln \left(\frac{1}{\text{P}}p_{\bf \atilde{Z}}
      \left(\frac{\vv}{\text{P}}\right) \right) \, d \vv - \ln \text{P} \nonumber\\
    &=& - \int p_{\aPw{\vN} {\bf \atilde{Z}}} ( \vv - \vx) \, \ln p_{\text{P} {\bf \atilde{Z}}} ( \vv) \, d \vv - \ln \text{P} \nonumber\\
    &=& D \left( p_{\aPw{\vN} {\bf \atilde{Z}}} (\vv - \vx) \big\| p_{\text{P} {\bf \atilde{Z}}} (\vv) \right) 
    + h(\aPw{\vN} {\bf \atilde{Z}}) - \ln \text{P}\nonumber\\
    &=& D \left(p_{\aPw{\vN} {\bf \atilde{Z}}} (\vv - \vx) \big\| p_{\text{P} {\bf \atilde{Z}}} (\vv) \right)  
    + h({\bf \atilde{Z}})+ \ln \frac{\aPw{\vN}}{ \text{P}} \label{divform},
  \end{eqnarray}
  where $D(p \| q)$ is the Kullback-Leibler divergence between two PDFs
  $p$ and $q$. Using equation~(\ref{divform}), the input cost
  constraint can be rewritten as:
  \begin{equation*}
    \Ep{\vX}{ D \left( p_{\aPw{\vN} {\bf \atilde{Z}}} (\vv - \vX) \big\| p_{\text{P} {\bf \atilde{Z}}}({\bf v})
      \right) }  = \ln \frac{ \text{P}}{ \aPw{\vN}}.
  \end{equation*}
\end{proof}

Note that the capacity problem at hand of the stable
channel~(\ref{chanstable}) under the input cost
constraint~(\ref{coscons}) is a generalization to the well known AWGN
channel under the average power constraint~\cite{Sha48_1} and the
additive independent Cauchy channel under a logarithmic
constraint~\cite{Fahs14-1}.

Finally, in the scalar case the generic cost function $\mathcal{C}(x)$
presented in~(\ref{cosfun}) is $\Theta(x^2)$ when $\alpha =2$. For all
other values of $\alpha$, using the same methodology as
in~\cite{Fahsarxiv} one can prove that $\mathcal{C}(x) =
\Theta\left(\ln|x|\right)$ by virtue of the fact that $\ln
p_{\atilde{Z}}(x) = \Theta(\ln |x|)$.  This comes in accordance with
the results presented in~\cite{Fahsarxiv}.

\subsection{Extensions and Insights}

The \apower measure $\aPw{\vX}$ defined in~(\ref{newpow}) is related
to a choice of ${\bf \atilde{Z}}$ --or equivalently a choice of $ 0 <
\alpha \leq 2$, and as previously mentioned $\aPw{\vX}$ can be looked
at as the relative power of $\vX$ with respect to that of ${\bf
  \atilde{Z}}$. Naturally one would ask the following: In a specific
scenario, what value of alpha is more suitable? An answer to this
question is given when considering, for example, an additive noise
channel $\vY = \vX + {\bf N}$. In fact, in most
communications' applications, the quantity of interest for a system
engineer is the received signal or the output $\vY$ as it
generally represents the quantity that will undergo further processing
in order to retrieve the useful information. In addition, the noise
variable ${\bf N}$ imposed by the channel represents another important
variable since relevant quantities and performance measures are
computed function of the relative power between the output signal and
the noise, a quantity that is commonly referred to as the output
SNR. Moreover, the output $\vY$ is shaped by the noise ${\bf N}$,
hence it has ``similar'' characteristics to those of ${\bf N}$ (for
example, a vector ${\bf N}$ having infinite variance components will
always induce a vector $\vY$ having infinite variance
components). This is to say, that in the context of an additive stable
noise channel, it would seem natural to measure the power of the
different signals with respect to a reference stable variable whose
power evaluates to unity. Hence the choice of $\alpha$ and then ${\bf
  \atilde{Z}}$ becomes straightforward depending on the stable noise
characteristic exponent $\alpha$.
  
A natural extension is to generalize the adoption of $\aPw{\vX}$ for a
specific ${\bf \atilde{Z}}$ to cases where the noise is not
necessarily stable but falls instead in the domain of normal
attraction $\mathbb{D}_{\alpha}$~\cite{kolmo,ibrlin} of the stable
variables.
For example, in the scalar case, any noise variable having a finite
second moment belongs to $\mathbb{D}_{2}$ and $\twoPw{X}$ is equal to
the second moment. For noise variables whose tail behavior is $\Theta
\left( \frac{1}{|x|^{1+\alpha}} \right)$, $0<\alpha<2$, the
  \apower $\aPw{X}$ should be used.


\section{\aFI: A Generalized Information Measure}
\label{sc:GFI}

In this section, we introduce a family of new information measures
$\left\{J_{\alpha}(\cdot)\right\}_{0 < \alpha \leq 2}$ and its
properties as a generalization of the standard Fisher
information. This is done through enforcing a family of identities of
the de Bruijn type and finding an analytical expression of
$J_{\alpha}(\cdot)$, $0 < \alpha \leq 2$.

\begin{definition}[\aFI]
  \label{def:Jalpha}
  Let $X$ be a finite differential entropy RV and $\atilde{Z}$ an
  independent reference SS variable $\atilde{Z} \sim
  \mathcal{S}\left(\alpha, \left( \frac{1}{\alpha}
    \right)^{\frac{1}{\alpha}} \right)$, $0 < \alpha \leq 2$. We
  define the ``Fisher information of order $\alpha$'' or the {\em
      \aFI} $J_{\alpha}(X)$ as follows:
  \begin{equation}
    J_{\alpha}(X)  = \alpha \cdot \lim_{\epsilon \rightarrow 0^{+}} \frac{h \left( X + \sqrt[\alpha]{\epsilon} \, \atilde{Z} \right) 
      - h(X)}{\epsilon}, \label{jdef_0}
  \end{equation}
  whenever the limit exists.

  For a $d$-dimensional random vector $\vX = (X_1, \cdots, X_d)$,
  $J_{\alpha}(\vX)$ is defined as in~(\ref{jdef_0}) where
  $\atilde{\vZ} \sim \bm{\mathcal{S}}\left(\alpha, \left(
      \frac{1}{\alpha} \right)^{\frac{1}{\alpha}} \right)$ is the
  $d$-dimensional reference SS vector.
\end{definition}

Alternatively, by the change of variable $t =
\frac{\epsilon}{\alpha}$, if $N$ denotes an independent SS variable $N
\sim \mathcal{S}\left(\alpha,1\right)$, the \aFI is
\begin{equation}
  J_{\alpha}(X)  = \lim_{t \rightarrow 0^{+}} \frac{h \left( X + \sqrt[\alpha]{t} \, N \right) - h(X)}{t}, \label{jdef}
\end{equation}
whenever the limit exists. In the vector case, $J_{\alpha}(\vX)$ is
also as in~(\ref{jdef}) where $\vN \sim \bm{\mathcal{S}}\left(\alpha,
  1\right)$.

Before proceeding to discuss the properties of the newly defined
quantity we point out that the existence of the limit is guaranteed in
a wide range of scenarios:
\begin{theorem}
  \label{exist-gfisher}
  For all random vectors $\vX$ such that $\E{\ln \left(1 + \|\vX\|
      \right)}$ and $h(\vX)$ are finite, $J_{\alpha}(\vX)$ exists for
    all $0 < \alpha \leq 2$.
\end{theorem}

\begin{proof}
  We first note that $h(\vX + \sqrt[\alpha]{t} \vN)$ exists and is
  finite since $\sqrt[\alpha]{t} \, \vN \sim \bm{\mathcal{S}} \left(
    \alpha, \sqrt[\alpha]{t} \right)$ has a bounded PDF and $\E{ \ln
    \left(1 + \|\vX\|\right) }$ is finite~\cite[Proposition
  1]{rioul2011}. Also, in the scalar case it has been proven
  in~\cite{FAF15-1} that the differential entropy $h(X+N_{t})$ is
  concave in $t \geq 0$ whenever $N$ is an infinitely divisible RV
  where $N_{t}$ is related to $N$ through their characteristic
  functions as follows~\cite[Theorem 2.3.9 p.65]{Hey04}:
  \begin{equation*}
    \phi_{N_{t}} (\omega)= e^{t \ln \phi_{N}(\omega)}.
  \end{equation*}
  Since in our case the infinitely divisible RV is stable with
  characteristic exponent $\alpha$, then $N_{t} \sim \sqrt[\alpha]{t}N
  $ which implies that $h(X + N_{t})$ is concave in $t$ and therefore
  it is everywhere left and right differentiable and a.e
  differentiable. These properties generalize in a straightforward
  manner to the vector case, and hence $\frac{d}{dt}h(\vX
  +\sqrt[\alpha]{t}\vN)$ exists a.e. in $t$ and
  $\frac{d}{dt}h(\vX+\sqrt[\alpha]{t}\vN)\Big|_{t=0^+}$ exists.
\end{proof}


\subsection{Properties of the \aFI}

Few properties of $J_{\alpha}(\vX)$ may be readily identified.
\begin{itemize}
\item[(1)] {\bf \em It is non-negative:\/} By definition,
  $J_{\alpha}(\vX)$ represents the rate of variation of $h(\vX)$ under
  a small disturbance in the direction of a standard SS vector. It
  represents the limit of positive quantities and therefore, {\bf
    $J_{\alpha}(\vX) \geq 0$}.

\item[(2)] {\em $J_{2}(\vX)$ coincides with the usual notion of Fisher
    information:\/} When the stable noise $\vN$ is Gaussian,
  i.e. $\alpha =2$, {\bf $J_{2}(\vX)$ coincides is the trace of the
    Fisher information matrix\/}.

\item[(3)] {\em It's translation invariant:\/} Let $\vc \in
  \Reals^{d}$, then {\bf $J_{\alpha}(\vX+\vc) =
    J_{\alpha}(\vX)$}. This follows directly from the definition and
  from the translation invariant property of the differential entropy.

\item[(4)] {\em It has a closed-form expression for symmetric stable
    vectors:\/} {\bf If $\vX \sim \bm{\mathcal{S}} \left( \alpha,
      \gamma \right)$ then $J_{\alpha}(\vX) = \frac{d}{\alpha}
    \frac{1}{\gamma^\alpha}$ nats}. Indeed, if $\vX \sim
  \bm{\mathcal{S}} \left( \alpha, \gamma \right)$ then $\vX +
  \sqrt[\alpha]{\epsilon} \vN \sim \bm{\mathcal{S}} \left( \alpha,
    \sqrt[\alpha]{\gamma^{\alpha} + \epsilon} \right)$ and
  \begin{eqnarray*}
    J_{\alpha}(\vX) &=& \lim_{\epsilon \rightarrow 0} \frac{h\left(\vX + \sqrt[\alpha]{\epsilon}\vN\right) - h(\vX)}{\epsilon}\\
    &=& \lim_{\epsilon \rightarrow 0} \frac{h\left(\sqrt[\alpha]{\gamma^\alpha+\epsilon}\vN\right) - h(\gamma \vN)}{\epsilon}\\
    &=& \lim_{\epsilon \rightarrow 0} \frac{h\left(\vN\right) + d \ln \left(\sqrt[\alpha]{1 + \frac{\epsilon}{\gamma^\alpha}}\right)
      - h(\vN)}{\epsilon}\\
    &=& \frac{d}{\alpha}\frac{1}{\gamma^\alpha} \,\,\text{nats}.
  \end{eqnarray*}
  
  This result comes in accordance with the fact that $J_2(\vX) =
  J(\vX) = \frac{d}{\sigma^2}$ whenever $\vX \sim
  \bm{\mathcal{N}}(0;\sigma^2 \mat{I})$ is Gaussian with covariance
  matrix $\sigma^2 \mat{I}$. This is true since in this case $\alpha =
  2$ and for a Gaussian variable $\gamma^2 = \frac{\sigma^2}{2}$.

\item[(5)] {\em Scales:\/}  $J_{\alpha}(a\vX) =
  \frac{1}{|a|^{\alpha}}J_{\alpha}(\vX)$ {\bf for} $a \neq 0$. Indeed,
  \begin{eqnarray*}
    J_{\alpha}(a\vX) &=& \lim_{\epsilon \rightarrow 0} \frac{h\left(a\vX + \sqrt[\alpha]{\epsilon}\vN\right) - h(a\vX)}{\epsilon}\\
    &=& \lim_{\epsilon \rightarrow 0} \frac{h\left(\vX + \sqrt[\alpha]{\frac{\epsilon}{|a|^{\alpha}}}\vN\right) + d \ln |a| - h(\vX) 
      - d \ln |a|}{\epsilon}\\
    &=& \frac{1}{|a|^{\alpha}} \lim_{\epsilon \rightarrow 0} \frac{h\left(\vX + \sqrt[\alpha]{\frac{\epsilon}{|a|^{\alpha}}}\vN\right) 
      - h(\vX)}{\frac{\epsilon}{|a|^{\alpha}}}\\
    &=& \frac{1}{|a|^{\alpha}} J_{\alpha}(\vX),
  \end{eqnarray*}
  where we used the fact that $(-\vN)$ is identically distributed as
  $\vN$.
      
\item[(6)] {\em Independent sums:\/} {\bf $J_{\alpha}(\vX + \vZ) \leq
    J_{\alpha}(\vX)$ when $\vZ$ is independent of $\vX$}. Indeed
  \begin{align*}
    J_{\alpha}(\vX + \vZ) &= \lim_{t \rightarrow 0} \frac{h\left(\vX + \vZ + \sqrt[\alpha]{t}\vN\right) - h(\vX + \vZ)}{t} \\
    &= \lim_{t \rightarrow 0} \frac{I(\vX+\vZ +  \sqrt[\alpha]{t}\vN;\vN)}{t} \\
    &\leq \lim_{t \rightarrow 0} \frac{I(\vX+ \sqrt[\alpha]{t}\vN;\vN)}{t} = J_{\alpha}(\vX),
  \end{align*}
  where the inequality is due to the fact that $\vN \mbox{ --- } \vX +
  \sqrt[\alpha]{t} \vN \mbox{ --- } \vX + \vZ + \sqrt[\alpha]{t}\vN$ is
  a Markov chain.

\item[(7)] {\em Sub-additivity:\/} {\bf $J_{\alpha}(\cdot)$ is
    sub-additive for independent random vectors\/}. Let $\vX =
  \left(X_1,\cdots,X_d\right)$ be a collection of $d$ independent RVs
  having Fisher information $\{J_{\alpha}(X_i)\}_{i=1}^{d}$, then $
  J_{\alpha}(\vX) = J_{\alpha}(X_1, \cdots, X_d) \leq \sum_{i=1}^{d}
  J_{\alpha}(X_i)$, because $h(Z_1, \cdots, Z_d) \leq \sum_{i=1}^{d}
  h(Z_i)$ with equality whenever $\{Z_i\}_{i=1}^d$ are independent.
  It is known that $J_2(\cdot)$ is additive and it will be later shown
  that {\bf $J_{\alpha}(\cdot)$ is in fact additive}.
\end{itemize}
  
Due to the above, one may consider $J_{\alpha}(\vX)$, $0 < \alpha \leq
2$ as a measure of information. A single random vector $\vX$ might
hence have different information measures which represent from an
estimation theory perspective a reasonable fact since the statistics
of the additive noise $\vN$ affect the estimation of $\vX$ based on
the observation of $\vX+\vN$. From this perspective, the original
Fisher information would seem suitable when the adopted noise model is
Gaussian or when we are restricting the RV to have a finite second
moment. 

\subsection{An expression of $J_{\alpha}(\cdot)$}
\label{sc:Main}

We find in what follows an expression of $J_{\alpha}(\vX)$
  whenever the random vector is absolutely continuous with a positive
  PDF. More precisely, let $\vX \in \set {V}$ where,
\begin{equation*}
  \set{V} = \left\{ \text{Absolutely continuous RVs } \vU: p_{\vU}(\vu) > 0, h(\vU) \text{ is finite } \& \int \ln\left(1 
      + \|\vU\|\right)\,p_{\vU}(\vu)\,d\vu \text{ is finite } \right\}.
\end{equation*}

\begin{lemma}[An Expression of the \aFI]
  \label{lem2}
  Let $\vN \sim \bm{\mathcal{S}} \left( \alpha, \gamma \right)$ be a
  SS vector and let $\vX \in \mathcal{V}$ be independent of $\vN$ with
  a characteristic function
  $\phi_{\vX}(\vomega)$ such that $\left[\|\vomega\|^{\alpha}
    \phi_{\vX}(-\vomega)\right] \in \text{L}^{1}(\Reals^{d})$. If
  there exists an $\epsilon > 0$ such that \footnote[5]{${\mathcal
      F}_{\mathcal{I}}(\cdot)$ denotes the inverse distributional
    Fourier transform. The regularity condition imposed
    in~(\ref{condthrough}) is assumed to hold whenever
    $J_{\alpha}(\cdot)$ is being evaluated using
      equation~(\ref{anadef}) throughout the paper.
  }~\begin{equation}  \bigg\{ \Bigl| \ln p_{\vX + \sbs
      \sqrt[\alpha]{t} \vN}(\vx) {\mathcal F}_{\mathcal{I}}\bigl[ \|\vomega\|^{\alpha}
    \phi_{\vX + \sbs \sqrt[\alpha]{t} \vN}(-\vomega)\bigr](\vx)
    \Bigr| \bigg\}_{\hspace{-0.08cm}t\in
      [0,\epsilon)} \label{condthrough}
  \end{equation}
  are uniformly bounded in $t$ by an integrable function of $\vx$,
  then the \aFI of $\vX$ is
  \begin{equation}
    \label{anadef}
    J_{\alpha}(\vX) = \int \ln p_{\vX}(\vx) \, {\mathcal F}_{\mathcal{I}}\bigl[ \|\vomega\|^{\alpha} \phi_{\vX}(-\vomega) \bigr] (\vx)\,d\vx.
  \end{equation}
\end{lemma}
 
\begin{proof}
  Using Theorem~\ref{exist-gfisher} $J_{\alpha}(\vX)$ exists. Now, let
  $t \geq \eta \geq 0$ and denote $ \vX_{t} = \vX + \sqrt[\alpha]{t}
  \vN$
  with characteristic function
  \begin{align*}
    \phi_{\vX_{t}}(\vomega) & = \phi_{\vX}(\vomega) \, e^{-t \gamma^{\alpha} \|\vomega\|^{\alpha}}
    \bs = \phi_{\vX_\eta}(\vomega) \, e^{-(t-\eta) \gamma^{\alpha} \|\vomega\|^{\alpha}} \\
    & = \phi_{\vX_\eta}(\vomega) - (t-\eta) \gamma^{\alpha}  \| \vomega \|^{\alpha} \phi_{\vX_\eta}(\vomega) + o(t-\eta).
  \end{align*}
  
  By the linearity of the inverse Fourier transform,
  \begin{equation}
    p_{\vX_{t}} \sbs (\vx) = p_{\vX_\eta} \sbs (\vx) - (t-\eta) \gamma^{\alpha} {\mathcal F}_{\mathcal{I}}\bigl[\| \vomega \|^{\alpha} 
    \phi_{\vX_\eta} \sbs (-\vomega)\bigr](\vx) + o(t-\eta),
    \label{derivp}
  \end{equation}
  which is valid since the inverse distributional Fourier transform
  ${\mathcal F}_{\mathcal{I}} \bigl[ \|\vomega\|^{m
    \alpha}\phi_{\vX_\eta}(\vomega) \bigr]$ exists for all $m \geq 1$
  because $\|\vomega\|^{m \alpha}\phi_{\vX_\eta}(\vomega)$ is a
  tempered function by virtue of the fact that
  $\phi_{\vX_\eta}(\vomega)$ is an $\text{L}^{1}$-characteristic
  function and hence is in $\set{L}^{\infty}(\Reals^d)$.
  Equation~(\ref{derivp}) implies that
  \begin{equation*}
    \frac{d\,p_{{\vX}_{\tau}}(\vx)}{d\tau} \bigg|_{\tau = \eta} =  - \gamma^{\alpha} {\mathcal F}_{\mathcal{I}}\bigl[ \|\vomega\|^{\alpha}
    \phi_{\vX_{\eta}}(-\vomega)\bigr](\vx),
  \end{equation*}
  and by the Mean Value Theorem, for some $0 \leq b(t) \leq t$,
  \begin{align*}
    & \frac{h({\vX}_{t}) - h({\vX})}{t} = - \int_{\Reals^d} \frac{p_{{\vX}_{t}}(\vx)\,\ln p_{{\vX}_{t}}(\vx) - p_{{\vX}}(\vx)\,
      \ln p_{{\vX}}(\vx)}{t} \,d\vx \\
    & = - \int_{\Reals^d}  \left[ 1+\ln p_{{\vX}_{b(t)}}(\vx) \right] \frac{d\,p_{{\vX}_{\tau}}(\vx)}{d\tau} \bigg|_{\tau = b(t)}\,d\vx \\
    & = \gamma^{\alpha} \int_{\Reals^d}  \left[ 1+\ln p_{{\vX}_{b(t)}}(\vx) \right] {\mathcal F}_{\mathcal{I}}\bigl[ \|\vomega\|^{\alpha}
    \phi_{\vX_{b(t)}}(-\vomega)\bigr](\vx) \,d\vx\\
    & = \gamma^{\alpha} \int_{\Reals^d} \ln p_{{\vX}_{b(t)}}(\vx)  {\mathcal F}_{\mathcal{I}}\bigl[ \|\vomega\|^{\alpha} 
    \phi_{\vX_{b(t)}}(-\vomega)\bigr](\vx) \,d\vx,
  \end{align*}
  which is true since $\left[\|\vomega\|^{\alpha} \phi_{\vX}
    (-\vomega) \right] \in \text{L}^{1} (\Reals^{d})$ and
  \begin{equation*}
    \int {\mathcal F}_{\mathcal{I}}\bigl[ \|\vomega\|^{\alpha} \phi_{\vX_{b(t)}}(-\vomega)\bigr](\vx) \, d\vx
    = \int \delta(\vomega) \, \|\vomega \|^{\alpha}  \phi_{\vX_{b(t)}}(-\vomega) \, d \vomega = 0.
  \end{equation*}

  The imposed conditions insure that Lebesgue's Dominated Convergence
  Theorem (DCT) holds and the limit may be passed inside the integral
  and
  \begin{equation*}
    J_{\alpha}(\vX) = \int_{\Reals^d} \ln p_{{\vX}}(\vx)  {\mathcal F}_{\mathcal{I}}\bigl[ \|\vomega\|^{\alpha} 
    \phi_{\vX}(-\vomega)\bigr](\vx) \,d\vx.
  \end{equation*}
  
\end{proof}

We note that, whenever $\alpha = 2$, equation~(\ref{anadef}) gives the
regular expression of the Fisher information. In fact, in the scalar
case
\begin{equation*}
  J(X) = J_2(X) = \int \ln p_X(x) \, {\mathcal F}_{\mathcal{I}} \bigl[ |\omega|^{2} \phi_X(-\omega) \bigr] (x)\,dx = 
  - \int \ln p_X(x) \, \frac{d^{2}}{dx^{2}}\,p_X(x) \, dx,
\end{equation*}
where the last equality is valid as long as $\ln p_X(x) \,
\frac{d}{dx}\,p_X(x)|^{+\infty}_{-\infty}$ vanishes. In the
$d$-dimensional case, $J_2(\vX)$ is also consistent with the regular
definition of the Fisher information being the trace of the Fisher
information matrix. The sufficient condition listed in the statement
of the lemma, is a technical condition involving ``fractional''
derivatives of the PDF $p_{\vX}(\vx)$. Whenever $\alpha = 2$, this
condition boils down to similar type of conditions imposed by
Kullback~\cite[pages 26-27]{kullback} to prove the well-known result
relating the second derivative of the divergence to the Fisher
information: a result that implies de Bruijn's identity at zero~(see
\cite{rioul2011}).

Let $\vX_{\eta} = \vX + \sqrt[\alpha]{\eta} \vN'$ for some $\eta > 0$
where $\vN' \sim \bm {\mathcal{S}}(\alpha,\gamma)$ independent of
$\vN$. In Appendix~\ref{appsuff} it is shown that the regularity
condition on~(\ref{condthrough}) is satisfied and therefore
\begin{equation*}
  \frac{d}{d t}\,h(\vX_{\eta}+\sqrt[\alpha]{t} \vN) \Big|_{t = 0^+} = \gamma^{\alpha} J_{\alpha}( \vX_{\eta}).
\end{equation*}

Since $\sqrt[\alpha]{\eta} \, \vN' + \sqrt[\alpha]{t} \, \vN$ is
distributed according to $\sqrt[\alpha]{\eta + t} \, \vN$, this
equation is equivalent to a generalized de Bruijn's identity stated in
the following theorem.

\begin{theorem}[Generalized de Bruijn's identity]
  \label{thm:debruijn}
  Let $\vX \in \mathcal{V}$ and define for $\eta > 0$ the random
  vector $\vX_{\eta} = \vX + \sqrt[\alpha]{\eta} \vN$. For any $\eta >
  0$,
  \begin{equation}
    \label{debruijngen1}
    \frac{d}{d \eta}\,h(\vX_{\eta}) = \gamma^{\alpha} J_{\alpha}(\vX_{\eta}),
  \end{equation}
  where $J_{\alpha}(\vX_{\eta})$ is given by equation~(\ref{anadef}).
  Additionally, whenever the regularity condition~(\ref{condthrough})
  is satisfied by $\vX$,
  \begin{equation}
    \label{debruijngen0}
    \frac{d}{d \eta}\,h(\vX_{\eta}) \Big|_{\eta = 0^+} = \gamma^{\alpha} J_{\alpha}(\vX),
  \end{equation}
  where $J_{\alpha}(\vX)$ is given by equation~(\ref{anadef}).
\end{theorem}

To compute $J_{\alpha}(\cdot)$, we use the fast Fourier transform
using {\em Matlab} by following a similar methodology as
in~\cite{MItt99}. We plot in Figure~\ref{fig:figureJeval} the
evaluation of $J_{\alpha}(\cdot)$ for a collection of alpha-stable
variables $X \sim \mathcal{S}\left(r,(r)^{-\frac{1}{r}}\right)$
parameterized by the characteristic exponent $r$. It is observed that
as the value of $r$ increases, $J_{\alpha}(X)$ increases. Furthermore
for fixed $r$, $J_{\alpha}(X)$ decreases with $\alpha$.
\begin{figure}[htb]
  \begin{center}
    \includegraphics[width=5.5in]{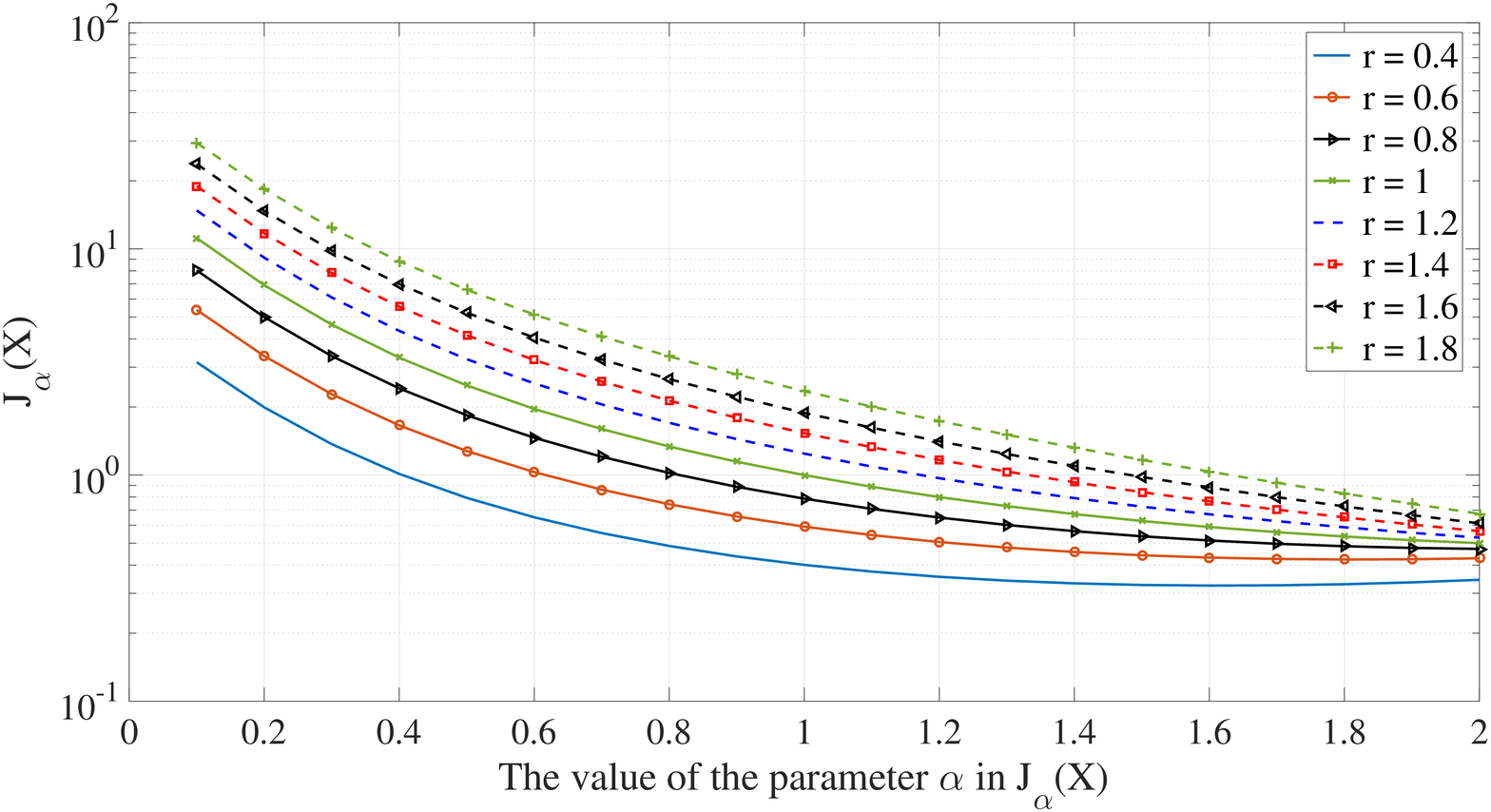}
    \caption{\small Evaluation of $J_{\alpha}(X)$ for $X \sim
      \mathcal{S}\left(r,(r)^{-\frac{1}{r}}\right)$ for different
      values of $\alpha$ and $r$}.
    \label{fig:figureJeval}
  \end{center}
\end{figure}

\section{Generalized Information Theoretic Identities}
\label{sec:GITI}

In addition to their theoretical relevance, information inequalities
have important implications in information theory. For example, by the
means of the FII, one can prove the EPI which is useful for finding
bounds on capacity regions and in proving strong versions of Central
Limit Theorems (CLT)s.  In what follows, we state and prove a list of
information inequalities featuring $J_{\alpha}(\cdot)$. Namely, we
list and prove a generalized FII, an upper bound on the differential
entropy of sums having a stable component and a generalized IIE.

\subsection{A Generalized Fisher Information Inequality}
\label{secgfii}

The Fisher information inequality is an important identity that
relates the Fisher information of the sum of independent RVs to those
of the individual variables. It was first proven by Stam~\cite{sta}
and then by Blachman~\cite{bla}. Both authors deduced the EPI from the
FII via de Bruijn's identity. Stam relied on a data processing
inequality of the Fisher information in the proof of the FII, a
methodology that was later used by Zamir~\cite{zam98} in a more
elaborate fashion. Finally, Rioul~\cite{rioul2011} derived a mutual
information inequality, an identity that implies the EPI and by the
means of de Bruijn's identity implies the FII.

\paragraph*{\underline{Data processing inequality for $J_{\alpha}$, $1
    < \alpha \leq 2$}}
The data processing inequality asserts that gains could not be
achieved when processing information. In terms of mutual information,
if the RVs $X \mbox{--} Y \mbox{--} Z$ form a Markov chain~\cite[p.34
Theorem 2.8.1]{cover},
\begin{equation*}
  I(Z;X) \leq I(Y;X),
\end{equation*}
with equality if $X \mbox{--} Z \mbox{--} Y$ is also a Markov
chain. In~\cite{zam98}, Zamir proved an equivalent inequality for the
Fisher information in a variable $Y$ about a parameter $\theta$. We
follow similar steps and extend the data processing inequality to
$J_{\alpha}$; an inequality which we will use next to prove the GFII.

\begin{definition}
  \label{defrtheta}
  Let $m > 0$ and let $\vtheta = \left[\theta_1 \, \theta_2 \, \cdots
    \, \theta_m\right]^{\text{t}}$ be a fixed vector of parameters.
  For $1 < \alpha \leq 2$ define,
  \begin{align}
    J_{\alpha}(\vY_{\vtheta};\vtheta) & \eqdef - \E{I^{\vtheta}_{2-\alpha}\left[\bigtriangleup_{\vtheta} 
        \ln p_{\vY_{\vtheta}}\right](\vY_{\vtheta})} \label{eqdeftheta}\\
    \& \qquad J_{\alpha}(\vY_{\vtheta};\vtheta|\vZ) & \eqdef \quad \Ep{\vZ}{J_{\alpha}\left(\vY_{\vtheta};\vtheta | 
        \vZ = \vz\right)} \nonumber,
  \end{align} 
  where for $f_{\vtheta}(\cdot): \Reals^{d} \rightarrow \Reals$ that
  is parameterized by $\vtheta \in \Reals^m$
  \begin{align}
    I^{\vtheta}_{2-\alpha}\left[f_{\vtheta}\right]: \, &\Reals^{d} \rightarrow \Reals \nonumber\\
    &\vx \,\,\,\rightarrow  I^{\vtheta}_{2-\alpha}\left[f_{\vtheta}\right](\vx) = \frac{\Gamma\left(\frac{m}{2}
        -\frac{2-\alpha}{2}\right)}{\pi^{\frac{m}{2}}2^{2-\alpha}\Gamma\left(\frac{2-\alpha}{2}\right)} \int_{\Reals^m} 
    \|\vtheta-\veta\|^{-m+2-\alpha}f_{\veta}(\vx)\,d\veta,  \label{Rtheta}
  \end{align}
  and 
  \begin{equation*}
    \bigtriangleup_{\vtheta}f_{\vtheta}(\vx) = \sum_{i =1}^{m} \frac{\partial^2 f_{\vtheta}}{\partial \theta^2_{i}}(\vx).
  \end{equation*}
\end{definition}

The operator $I_{2-\alpha}[\cdot]$ is the Riesz potential of order
$(2-\alpha)$ presented in Appendix~\ref{asd}. Note that the Riesz
  potential in equation~(\ref{Rtheta}) is that of function
  $f_{\vtheta}(\cdot)$ when $\vtheta$ is considered the variable
  instead of $\vx$.

\begin{theorem}[Translation Property for $d=m$]
  \label{th:transprop}
  
  If $d=m$ and $\vY_{\vtheta} = \vY +\vtheta$, then
  \begin{equation}
    J_{\alpha}(\vY_{\vtheta};\vtheta) = J_{\alpha}(\vY) .
    \label{locpar}
  \end{equation}
\end{theorem}

\begin{proof}
  \begin{align}
    J_{\alpha}(\vY)  &= \int \ln p_\vY(\vy) \,\mathcal{F}_{\mathcal{I}}\bigl[\|\vomega\|^{\alpha}\phi_\vY(- \vomega)
    \bigr](\vy)\,d\vy\nonumber\\
    & = -\int \ln p_\vY(\vy) \, \bigtriangleup_{\vy} \biggl(\mathcal{F}_{\mathcal{I}}\bigl[\|\vomega\|^{\alpha-2}\phi_\vY(- \vomega)
    \bigr](\vy)\biggr)\,d\vy\label{Fprop}\\
    & = -\int \bigtriangleup_{\vy} \left(\ln p_\vY(\vy)\right) \, I_{2-\alpha}[p_\vY](\vy) \,d\vy \label{parts}\\
    &= -\int I_{2-\alpha}\left[\bigtriangleup_{\vy} \left(\ln p_\vY\right)\right](\vy) \, p_\vY(\vy) \,d\vy \label{appjust}\\
    & = -\int  I_{2-\alpha}\left[ \sum_j \frac{d^2}{d y_j^2} \ln p_{\vY} \right](\vy)  \,p_\vY(\vy)\,d\vy \label{forlater}\\
    & = -\int I^{\vtheta}_{2-\alpha}\left[ \sum_j \frac{d^2}{d \theta_j^2} \ln p_{\vY_{\vtheta}} \right](\vy +\vtheta) 
    \,p_\vY(\vy) \,d\vy \label{sum}\\
    & = -\int  I^{\vtheta}_{2-\alpha}\left[ \bigtriangleup_{\vtheta} \ln p_{\vY_{\vtheta}} \right](\vy +\vtheta)  
    \,p_{\vY_{\vtheta}}(\vy+\vtheta)\,d\vy \nonumber\\
    & = J_{\alpha}(\vY_{\theta};\theta) \nonumber,
  \end{align}
  Equation~(\ref{Fprop}) is due to basic properties of the Fourier
  transform since $I_{2-\alpha}(p_\vY)(\vy) =\mathcal{F}_{\mathcal{I}}
  \left[\|\vomega\|^{\alpha-2}\phi_\vY(- \vomega) \right](\vy)$ decays
  to $0$ at ``$\infty$''. In order to write equation~(\ref{parts}), we
  use Green's first identity\cite{walt} in the following form: Let
  $\bigtriangledown$ denotes the gradient operator and $\times$
  denotes the dot product. If $\Psi(\cdot)$ and $\Phi(\cdot)$ are real
  valued functions on $\Reals^d$, then
  \begin{equation*}
    \int_{\Reals^d}\Psi(\vy) \bigtriangleup \Phi (\vy) \, d\vy = -\int_{\Reals^d} \bigtriangledown \Psi (\vy) 
    \times \bigtriangledown \Phi (\vy)\,d\vy + \lim_{R \rightarrow +\infty} \int_{\|\vy\| = R} \Psi(\vy) 
    \bigtriangledown \Phi (\vy) \times \vn \,dS(\vy),
  \end{equation*}
  where $\vn$ is the outward pointing unit normal vector of surface
  element $dS(\vy)$. Applying twice Green's theorem justifies
  equation~(\ref{parts}) as long as:
  \begin{eqnarray*}
    &\lim_{R \rightarrow +\infty}& \int_{\|\vy\| = R} \ln p_\vY(\vy) \bigtriangledown I_{2-\alpha}(p_\vY)(\vy) \times \vn \,dS(\vy)= 0\\
    &\text{and}&\\
    &\lim_{R \rightarrow +\infty}& \int_{\|\vy\| = R} I_{2-\alpha}(p_\vY)(\vy) \bigtriangledown \ln p_\vY(\vy) \times \vn \, dS(\vy)= 0.
  \end{eqnarray*}
  As stated in Appendix~\ref{asd}, equation~(\ref{appjust}) holds true
  whenever $\left|\bigtriangleup_{\vy} \ln p_{\vY}(\vy)\right| \,
  I_{2-\alpha}(p_{\vY})(\vy)$ is integrable. It remains to justify
  equation~(\ref{sum}) which we prove next,
  \begin{align}
    &I^{\vtheta}_{2-\alpha}\left[\sum_j \frac{d^2}{d \theta_j^2} \ln p_{\vY_{\vtheta}} \right](\vy+\vtheta)\nonumber\\
    &=  \frac{\Gamma\left(\frac{d}{2}-\frac{2-\alpha}{2}\right)}{\pi^{\frac{d}{2}}2^{2-\alpha}\Gamma\left(\frac{2-\alpha}{2}
      \right)} \int_{\Reals^d} \|\vtheta-\veta\|^{-d+2-\alpha}\sum_j \frac{d^2}{d \eta_j^2} 
    \ln p_{\vY_{\veta}}(\vy + \vtheta) \,d\veta \label{defrt}\\
    &=  \frac{\Gamma\left(\frac{d}{2}-\frac{2-\alpha}{2}\right)}{\pi^{\frac{d}{2}}2^{2-\alpha}\Gamma\left(\frac{2-\alpha}{2}
      \right)} \int_{\Reals^d} \|\vtheta-\veta\|^{-d+2-\alpha}\sum_j \frac{d^2}{d \theta_j^2} \ln p_{\vY}(\vy + \vtheta -\veta) 
    \,d\veta \label{transprop}\\
    &=  \frac{\Gamma\left(\frac{d}{2}-\frac{2-\alpha}{2}\right)}{\pi^{\frac{d}{2}}2^{2-\alpha}\Gamma\left(\frac{2-\alpha}{2}
      \right)} \int_{\Reals^d} \|\vtheta-\veta\|^{-d+2-\alpha}\sum_j \frac{d^2}{d y_j^2} \ln p_{\vY}(\vy + \vtheta -\veta) 
    \,d\veta \nonumber\\
    &=  \frac{\Gamma\left(\frac{d}{2}-\frac{2-\alpha}{2}\right)}{\pi^{\frac{d}{2}}2^{2-\alpha}\Gamma\left(
        \frac{2-\alpha}{2}\right)} \int_{\Reals^d} \|\vy - \vv\|^{-d+2-\alpha}\sum_j \frac{d^2}{d y_j^2} 
    \ln p_{\vY}(\vv) \,d\vv \label{chovar}\\
    &= I^{\vtheta}_{2-\alpha}\left[\sum_j \frac{d^2}{d y_j^2}  \ln p_{\vY} \right](\vy). \nonumber
  \end{align} 
  Equation~(\ref{defrt}) is the definition of $I^{\vtheta}_{2-\alpha}
  \left[ \cdot \right]$ given in equation~(\ref{Rtheta})
  and~(\ref{transprop}) is due to the fact that $\vY_{\veta} = \vY +
  \veta$. Equation~(\ref{chovar}) is obtained by the change of
  variable $\vv = \vy + \vtheta - \veta$ and the last equation is due
  to the definition of $I_{2-\alpha}\left[\cdot\right]$ (see
  Appendix~\ref{asd}).
\end{proof}

\begin{theorem}[Chain Rule and Data Processing Inequality for the
  \aFI]
  \label{dpif}

  If $\vtheta \text{--} \vY_{\vtheta} \text{--} \vZ_{\vtheta}$, i.e.,
  the conditional distribution of $\vZ_{\vtheta}$ given
  $\vY_{\vtheta}$ is independent of $\vtheta$, then
  \begin{equation*}
    J_{\alpha}(\vZ_{\vtheta};\vtheta) \leq J_{\alpha}(\vY_{\vtheta};\vtheta),
  \end{equation*}
  whenever $J_{\alpha}(\vY_{\vtheta};\vtheta|\vZ_{\vtheta}) \geq 0$.
\end{theorem}

We note that the condition $J_{\alpha} ( \vY_{\vtheta}; \vtheta |
\vZ_{\vtheta} ) \geq 0$ is needed since there are no formal guarantees
of non-negativeness according to Definition~\ref{defrtheta} as it is
the case for $J_{\alpha}(\vY)$. The non-negativity of
$J_{\alpha}(\vY_{\vtheta};\vtheta)$ is guaranteed, for example,
whenever $\vtheta$ is a translation parameter. Another case when
non-negativity holds is found next in the proof of
Theorem~\ref{GFIItheorem}.

\begin{proof}
  Consider
  \begin{equation*}
    J_{\alpha}(\vY_{\vtheta},\vZ_{\vtheta};\vtheta) = - \Ep{\vY,\vZ}{I_{2-\alpha}\left[\bigtriangleup_{\vtheta}
        \left(\ln p_{\vY_{\vtheta},\vZ_{\vtheta}}\right)\right] (\vY_{\vtheta},\vZ_{\vtheta})}. 
  \end{equation*} 
  We have 
  \begin{equation*}
    \ln p_{\vY_{\vtheta},\vZ_{\vtheta}}(\vy,\vz;\vtheta) = \ln p_{\vZ_{\vtheta}}(\vz;\vtheta) + \ln p_{\vY_{\vtheta}|\vZ_{\vtheta}}
    (\vy;\vtheta | \vz),
  \end{equation*}
  which yields
  \begin{eqnarray}
    J_{\alpha}(\vY_{\theta},\vZ_{\theta};\theta) &=&  J_{\alpha}(\vZ_{\theta};\theta) + J_{\alpha}(\vY_{\theta};\theta|\vZ_{\theta}) 
    \label{chainrule}\\
    &\geq&  J_{\alpha}(\vZ_{\theta};\theta) \label{refbd}.
  \end{eqnarray}
  Equation~(\ref{chainrule}) is due to the linearity property of the
  Laplacian operator, the Riesz potential~\cite{stein70} and the
  expectation operator.  Equation~(\ref{refbd}) is justified by the
  fact that
  $J_{\alpha}(\vY_{\theta};\theta|\vZ_{\theta}) \geq 0$ by
  assumption. Equality holds if and only if $J_{\alpha}(\vY_{\theta};
  \theta|\vZ_{\theta}) = 0$ which is true if $\theta \text{--}
  \vZ_{\theta} \text{--} \vY_{\theta}$ forms a Markov chain. On the
  other hand, since $\vZ_{\theta}$ is conditionally independent of
  $\theta$ given $\vY_{\theta}$, $\ln p_{\vZ_{\theta}|\vY_{\theta}} (
  \cdot | \vy )$ is independent of $\theta$ and
  \begin{equation*}
    J_{\alpha}(\vY_{\theta},\vZ_{\theta};\theta) = J_{\alpha}(\vY_{\theta};\theta),
  \end{equation*}
  which along with equation~(\ref{refbd}) gives the required result.
\end{proof}

\paragraph*{\underline{Additivity property of $J_{\alpha}(\vY)$ for
    vectors $\vY$ having independent components}}
Before proceeding to state and prove the GFII, we prove the {\em
  additivity of $J_{\alpha}(\vY)$ when $\vY$ has independent
  components\/}, as mentioned in property (7). Starting from
equation~(\ref{forlater}),
\begin{eqnarray}
  J_{\alpha}(\vY) &=&  -\int p_\vY(\vy)\, I_{2-\alpha}\left[ \sum_j \frac{d^2}{d y_j^2} \ln p_{\vY} \right](\vy)  \,d\vy \nonumber\\
  &=& -\int p_\vY(\vy)\, I_{2-\alpha}\left[ \sum_j \frac{d^2}{d y_j^2} \ln p_{Y_{j}} \right](\vy)  \,d\vy \label{indcomp}\\
  &=& -\sum_j \int p_\vY(\vy)\, I_{2-\alpha}\left[ \frac{d^2}{d y_j^2} \ln p_{Y_{j}} \right](y_j)  \,d\vy \label{linearriesz}\\
  &=&  -\sum_j \int p_{Y_j}(y_j)\, I_{2-\alpha}\left[ \frac{d^2}{d y_j^2} \ln p_{Y_{j}} \right](y_j)  \,dy_j \label{indcomp1}\\
  &=& -\sum_j \int  \frac{d^2}{d y_j^2}I_{2-\alpha}\left[p_\vY(\vy)\right]\, \ln p_{Y_{j}} (y_j)  \,dy_j \label{reggie}\\
  &=& \sum_j J_{\alpha}(Y_j)\nonumber,
\end{eqnarray}
where equations~(\ref{indcomp}) and~(\ref{indcomp1}) are due to the
independence of the $Y_{j}$'s. Equation~(\ref{linearriesz}) is
justified by the linearity of the Riesz potential and
equation~(\ref{reggie}) holds true whenever $\left\{\ln p_{Y_{j}}
  \frac{d}{dy_j}I_{2-\alpha}[p_{Y_{j}}](y_j)\right\}_{j}$ go to $0$ at
``$\infty$'' and the regularity condition~(\ref{condthrough}) is
satisfied by the $\left\{Y_j\right\}$'s.

\paragraph*{\underline{Generalized Fisher Information Inequality}}

\begin{theorem}[Generalized Fisher Information Inequality (GFII)]
  \label{GFIItheorem}

  Let $1 < \alpha \leq 2$ and let $\vY_1$ and $\vY_2$ be two
  independent $d$-dimensional random vectors, then
  \begin{equation}
    \label{gfii}
    J^{\frac{1}{1-\alpha}}_{\alpha}(\vY_1+\vY_2) \geq J^{\frac{1}{1-\alpha}}_{\alpha}(\vY_1)  
    + J^{\frac{1}{1-\alpha}}_{\alpha}(\vY_2).
  \end{equation}
\end{theorem}

We note that whenever $\alpha = 2$, equation~(\ref{gfii}) boils down
to the well-known ``classical'' FII.

\begin{proof}
  For the matter of the proof, we make use of the data processing
  inequality established in Theorem~\ref{dpif}. Let $\omega_1$ and
  $\omega_2 \in \Reals^{+*}$ be two positive numbers such that
  $\omega_1 + \omega_2 = 1$. Also let $\epsilon > 0$ and ${\bf N}$ be
  an independent random vector distributed according to $\bm
  {\mathcal{S}}(\alpha,1)$. For any $\vtheta \in \Reals^{d}$ we have
  \begin{equation*}
    \vtheta - \left(\frac{\vY_1}{\omega_1} +\vtheta ,\frac{\vY_2}{\omega_2}+\vtheta \right) - (\vY_1 + \vY_2 
    + \vtheta + \sqrt[\alpha]{\epsilon}{\bf N})
  \end{equation*} 
  forms a Markov chain.  Define $\vY_{\vtheta,1} =
  \frac{\vY_1}{\omega_1} + \vtheta$, $\vY_{\vtheta,2} =
  \frac{\vY_2}{\omega_2} + \vtheta$ and $\vZ_{\vtheta} = \omega_1
  \vY_{\vtheta,1} + \omega_{2} \vY_{\vtheta,2} +
  \sqrt[\alpha]{\epsilon}\vN$, then
  \begin{equation*}  
    J_{\alpha}\left(\left(\vY_{\vtheta,1},\vY_{\vtheta,2}\right);\vtheta \Big|\vZ_{\vtheta}\right) \eqdef \Ep{\vZ_{\vtheta}}
    {J_{\alpha}\left(\left(\vY_{\vtheta,1},\vY_{\vtheta,2}\right);\vtheta | \vZ_{\vtheta} \right)} \geq 0.
  \end{equation*}
  
  Indeed, let $p_{\left(\vY_{\vtheta,1}, \vY_{\vtheta,2} \right)
    |\vZ_{\vtheta}}(\cdot,\cdot |\vz)$ be the PDF of
  $\left(\vY_{\vtheta,1},\vY_{\vtheta,2}\right)$ given $\vZ_{\vtheta}
  = \vz$. Then,
  \begin{eqnarray*}
    p_{\left(\vY_{\vtheta,1},\vY_{\vtheta,2}\right) |\vZ_{\vtheta}}(\vy_{1},\vy_2 |\vz)
    &=& p_{\vY_{\vtheta,1} |\vZ_{\vtheta}}(\vy_1|\vz)\,p_{\vY_{\vtheta,2} |\vY_{\vtheta,1},\vZ_{\vtheta}}(\vy_2 |\vy_1,\vz)\\
    &=& p_{\frac{\vY_{1}}{\omega_1} |\vZ_{\vtheta}}(\vy_1-\vtheta|\vz)\,p_{\sqrt[\alpha]{\epsilon}{\bf N}|\vY_{\vtheta,1},\vZ_{\vtheta}}
    (\vz-\omega_1\vy_1-\omega_2\vy_2 |\vy_1,\vz).
  \end{eqnarray*} 
  
  One can now write:
  \begin{eqnarray}
    J_{\alpha}\left(\left(\vY_{\vtheta,1},\vY_{\vtheta,2}\right);\vtheta \Big|\vZ_{\vtheta} = \vz\right)  &=&   
    J_{\alpha}\left(\vY_{\vtheta,1};\vtheta \Big|\vZ_{\vtheta}=\vz\right)  +  J_{\alpha}\left(\vY_{\vtheta,2};\vtheta 
      \Big|\left(\vY_{\vtheta,1},\vZ_{\vtheta}=\vz\right)\right) \nonumber \\
    &=& J_{\alpha}\left(\frac{\vY_{1}}{\omega_1}\Big|\vZ_{\vtheta}=\vz\right)  +  J_{\alpha}\left(\vY_{\vtheta,2};\vtheta 
      \Big|\left(\vY_{\vtheta,1},\vZ_{\vtheta}=\vz\right)\right)  \nonumber\\
    &=& J_{\alpha}\left(\frac{\vY_{1}}{\omega_1}\Big|\vZ_{\vtheta}=\vz\right)  \label{chainrule1},
  \end{eqnarray}   
  where we used Theorem~\ref{th:transprop} and the fact that
  \begin{equation*}
    J_{\alpha}\left(\vY_{\vtheta,2};\vtheta \Big|\left(\vY_{\vtheta,1},\vZ_{\vtheta}=\vz\right)\right) = \Ep{\vY_{\vtheta,1}}
    { J_{\alpha} \left(\vY_{\vtheta,2}; \vtheta \Big| \left(\vY_{\vtheta,1}, \vZ_{\vtheta} =\vz \right) \right)} = 0,
  \end{equation*} 
  since $J_{\alpha}\left(\vY_{\vtheta,2}; \vtheta \Big| \left(
      \vY_{\vtheta,1} = \vy_1, \vZ_{\vtheta}=\vz \right)
  \right) = 0$ for every $\vy_1$ because $p_{\vY_{\vtheta,2}
    |\left( \vY_{\vtheta,1}, \vZ_{\vtheta} \right)}(\cdot)$ is
  independent of $\vtheta$.
  Equation~(\ref{chainrule1}) is non-negative by property (1) and
  therefore by Theorem~\ref{dpif},
  \begin{eqnarray}
    J_{\alpha}\left(\vZ_{\vtheta};\vtheta\right) &\leq& J_{\alpha}\left(\left(\vY_{\vtheta,1},\vY_{\vtheta,2}\right);\vtheta\right).
    \label{fiiprim}
  \end{eqnarray}
  
  Since $\vY_{\vtheta,1}$ and $\vY_{\vtheta,2}$ are statistically
  independent and using the definition of $J_{\alpha}(\cdot;\vtheta)$
  in~(\ref{eqdeftheta}), the RHS of equation~(\ref{fiiprim}) boils
  down to:
  \begin{equation*}
    J_{\alpha}\left(\left(\vY_{\vtheta,1},\vY_{\vtheta,2}\right);\vtheta\right) =  J_{\alpha}\left(\vY_{\vtheta,1};\vtheta\right) 
    + J_{\alpha}\left(\vY_{\vtheta,2};\vtheta\right),
  \end{equation*}    
  which implies by means of the translation invariance property (2)
  in~(\ref{locpar}) that equation~(\ref{fiiprim}) is equivalent to:
  \begin{equation*}
    J_{\alpha}\left(\vY_1 + \vY_2 + \sqrt[\alpha]{\epsilon}\vN\right).
    \leq J_{\alpha}\left(\frac{\vY_1}{\omega_1}\right) + J_{\alpha}\left(\frac{\vY_2}{\omega_2} \right),
  \end{equation*}
  Under the regularity condition~(\ref{condthrough}), taking the limit
  as $\epsilon \rightarrow 0$ yields
  \begin{eqnarray}
    J_{\alpha}\left(\vY_1 + \vY_2\right) &\leq& J_{\alpha}\left(\frac{\vY_1}{\omega_1}\right) + J_{\alpha}\left(\frac{\vY_2}
      {\omega_2}\right) \nonumber\\
    &\leq& \omega_1^{\alpha} J_{\alpha}(\vY_1) + \omega_2^{\alpha} J_{\alpha}(\vY_2) \label{minpro},
  \end{eqnarray}
  by property (5) of $J_{\alpha}(\cdot)$.  Equation~(\ref{minpro})
  holds true for any $\omega_1$ and $\omega_2$ satisfying the
  conditions of the theorem, the tightest choice $\omega_1^*$ and
    $\omega_2^*$ being,
  \begin{align*}
    \omega_1^* & = \arg\min_{0 \leq \omega_1 \leq 1} \, \{\omega_{1}^{\alpha} J_{\alpha}(\vY_1) + (1-\omega_1)^{\alpha} J_{\alpha}(\vY_2)\} \\
    & = \frac{J^{\frac{1}{\alpha-1}}_{\alpha}(\vY_2)}{J^{\frac{1}{\alpha-1}}_{\alpha}(\vY_1) + J^{\frac{1}{\alpha-1}}_{\alpha}(\vY_2)} \\
    \omega_2^* & = 1 - \omega_1^* = \frac{J^{\frac{1}{\alpha-1}}_{\alpha}(\vY_1)}{J^{\frac{1}{\alpha-1}}_{\alpha}(\vY_1) + 
      J^{\frac{1}{\alpha-1}}_{\alpha}(\vY_2)}, 
  \end{align*}
  for which~(\ref{minpro}) becomes
  \begin{equation*}
    J_{\alpha}\left(\vY_1 + \vY_2\right)  
    \leq \frac{J_{\alpha}(\vY_1)J_{\alpha}(\vY_2)}{\left[J^{\frac{1}{\alpha-1}}_{\alpha}(\vY_1) + J^{\frac{1}{\alpha-1}}_{\alpha}
        (\vY_2)\right]^{\alpha-1}},
  \end{equation*}
  which completes the proof of the theorem.
\end{proof}

\subsection{Upper Bounds on the Differential Entropy of Sums Having a
  Stable Component}

An important category of information inequalities consists of finding
upper bounds on the entropy of independent sums.  Starting with
fundamental inequalities such that the upper bound on the discrete
entropy of independent sums~\cite{cover} and the upper bound on the
differential entropy of the sum of independent finite-variance
RVs~\cite{Sha48_1}, several identities involving discrete and
differential entropy of sums were subsequently shown
in~\cite{ruz,tao,kont1,mad08,cov94,ord06,bob13}. Recently
in~\cite{FAF15-1}, an upper bound on the differential entropy of the
sum $X + N$ of two independent RVs was found where $N$ is a
finite-variance infinitely divisible variable having a Gaussian
component. We extend in this section the known upper bound results to
cases where $\vN$ is SS stable vector using the GFII and the
generalized de~Bruijn's identity.

\begin{theorem}[Upper bound on the Entropy of Sums having a Stable
  Component]
  \label{bdsta}
  
  Let $\vZ \sim \bm{\mathcal{S}} \left( \alpha, \gamma \right)$, $1
  <\alpha \leq 2$, and let $\vX$ be a $d$-dimensional vector that is
  independent of $\vZ$ such that $h(\vX)$ and $J_{\alpha}(\vX)$ are
  finite. Then
  \begin{equation*}
    h(\vX+\vZ) - h(\vX) \leq \gamma^{\alpha} J_{\alpha}(\vX) \, _{2}F_{1}\left(\alpha-1,\alpha-1;\alpha;
      -\left(\frac{\alpha\gamma^{\alpha}}{d} J_{\alpha}(\vX)\right)^{\frac{1}{\alpha-1}}\right),
  \end{equation*}
  where $_{2}F_{1}(a,b;c;z)$ is the analytic continuation of the Gauss
  hypergeometric function on the complex plane with a cut along the
  real axis from 1 to +$\infty$.
\end{theorem}

For more details on hypergeometric functions, the reader may refer to
Appendix~\ref{asd}.  Theorem~\ref{bdsta} provides an upper bound on
the entropy of the sum of two variables when one of them is stable. As
a special case, when $\alpha =2$, it recovers the upper bound for
Gaussian noise setups~\cite{FAF15-1}.

\begin{proof}
  Using the extended de Bruijn's identity~(\ref{debruijngen1}), we
  write:
  \begin{align}
    h(\vX + \vZ) - h(\vX) & = \int^{1}_{0} \gamma^{\alpha} J_{\alpha}(\vX + \sqrt[\alpha]{\eta} \vZ) \,d\eta \nonumber\\
    &\leq \gamma^{\alpha} \int^{1}_{0} \frac{J_{\alpha}(\vX)J_{\alpha}( \sqrt[\alpha]{\eta} \vZ)}{\left(J^{\frac{1}{\alpha-1}}_{\alpha}(\vX) 
        + J^{\frac{1}{\alpha-1}}_{\alpha}(\sqrt[\alpha]{\eta}  \vZ)\right)^{\alpha -1}} \, d\eta \label{jgfii}\\
    &=   \gamma^{\alpha} \int^{1}_{0} \frac{J_{\alpha}(\vX)\frac{d}{\alpha \gamma^{\alpha} \eta }}{\left(J^{\frac{1}{\alpha-1}}_{\alpha}(\vX) 
        + (\frac{d}{\alpha \gamma^{\alpha} \eta })^{\frac{1}{\alpha-1}}\right)^{\alpha -1}} \, d\eta \label{propj}\\
    &= (\alpha-1)\gamma^{\alpha} J_{\alpha}(\vX) \int_{0}^{1} \frac{u^{\alpha-2}}{\left((\frac{\alpha \gamma^{\alpha}}{d}
        J_{\alpha}(\vX))^{\frac{1}{\alpha-1}}u + 1 \right)^{\alpha -1}} du \nonumber\\
    &=\gamma^{\alpha} J_{\alpha}(\vX) \, _{2}F_{1}\left(\alpha-1,\alpha-1;\alpha;-\left(\frac{\alpha\gamma^{\alpha}}{d}
        J_{\alpha}(\vX)\right)^{\frac{1}{\alpha-1}}\right) \label{uppsta},
  \end{align}
  where we use the GFII in order to write equation~(\ref{jgfii}) and
  properties (4) and (5) of $J_{\alpha}(\cdot)$ to validate
  equation~(\ref{propj}).
\end{proof}

Interestingly, Theorem~\ref{bdsta} gives an analytical bound on the
change in the transmission rates of the linear stable channel function
of an input scaling operation: let $a \neq 0$, then
\begin{align*}
  h(a\vX + \vZ)
  &\leq h(a \vX) + \gamma^{\alpha} J_{\alpha}(a\vX) \, _{2}F_{1}\left(\alpha-1,\alpha-1;\alpha;-\left(\frac{\alpha\gamma^{\alpha}}{d}
      J_{\alpha}(a\vX)\right)^{\frac{1}{\alpha-1}}\right),\\
  &= h(\vX) + d \ln |a| + \left(\frac{\gamma}{|a|}\right)^{\alpha} J_{\alpha}(\vX) \, _{2}F_{1}\left(\alpha-1,\alpha-1;\alpha;
    -\left(\frac{\alpha}{d}\left(\frac{\gamma}{|a|}\right)^{\alpha} J_{\alpha}(\vX)\right)^{\frac{1}{\alpha-1}}\right),
\end{align*}
where we used the fact that $h(a \vX) = h(\vX) + d \ln |a|$ and
$J_{\alpha}(a \vX) = \frac{1}{|a|^{\alpha}}
J_{\alpha}(\vX)$. Subtracting $h(\vZ)$ from both sides of the equation
gives
\begin{equation*}
  I(a\vX+\vZ;\vX) - I(\vX+\vZ;\vX) \leq  \ln |a| + \left(\frac{\gamma}{|a|}\right)^{\alpha} J_{\alpha}(\vX) 
  \, _{2}F_{1}\left(\alpha-1,\alpha-1;\alpha;-\left(\frac{\alpha}{d}\left(\frac{\gamma}{a}\right)^{\alpha}
      J_{\alpha}(\vX)\right)^{\frac{1}{\alpha-1}}\right).
\end{equation*}

Since $_{2}F_{1}\left(\alpha-1,\alpha-1;\alpha;0\right) = 1$,
\begin{equation*}
  \lim_{|a| \rightarrow +\infty}\left(\frac{\gamma}{a}\right)^{\alpha} J_{\alpha}(\vX) \, _{2}F_{1}\left(\alpha-1,\alpha-1;
    \alpha;-\left(\frac{\alpha}{d}\left(\frac{\gamma}{a}\right)^{\alpha}
      J_{\alpha}(\vX)\right)^{\frac{1}{\alpha-1}}\right) = 0,
\end{equation*}
and for large values of $|a|$ the variation in the transmissions rates
is bounded by a logarithmically growing function of $|a|$.  This is a
known behavior of the optimal transmission rates that are achieved by
Gaussian inputs in a Gaussian setting.

On a final note, making use of the identity:
\begin{equation*}
  \ln(1+t)  = t _2F_1(1,1;2;-t),
\end{equation*}
equation~(\ref{uppsta}) when evaluated for $\vZ \sim
\bm{\mathcal{N}}(0;\sigma^2 \mat{I})$ and $\alpha = 2$ boils down to
the following:
\begin{corollary}[Upper bound on the Entropy of Sums having a Gaussian
  Component]~\cite{FAF15-1}
  \label{newbdgaussian}

  Let $\vZ \sim \bm{\mathcal{N}}(0,\sigma^2 \mat{I})$ and $\vX$ be a
  $d$-dimensional vector that is independent of $\vZ$ such that
  $h(\vX)$ and $J(\vX)$ are finite. The differential entropy of $\vX +
  \vZ$ is upper
  bounded~by: 
  \begin{equation}
    \label{bdnew}
    h(\vX+\vZ) \leq h(\vX) + \frac{d}{2}\ln \left(1 + \frac{\sigma^2}{d} J(\vX)\right),
  \end{equation}
  and equality holds if and only if both $X$ and $Z$ are Gaussian
  distributed.
\end{corollary}
 
As shown in~\cite[Section 4]{FAF15-1}, we note that~(\ref{bdnew})
implies a reverse EPI when one of the vectors is Gaussian distributed,
which is equivalent to the concavity of the entropy power proved by
Costa~\cite{costa}. This was noted by
Courtade~\cite{courtadearxiv2016} who provided a generalization of the
reverse EPI in~\cite[Theorem 5]{courtade2016}.

\subsection{A Generalized Isoperimetric Inequality for Entropies}

Let ${\bf \atilde{Z}} \sim \bm{\mathcal{S}} \left( \alpha,
  (\frac{1}{\alpha})^{\frac{1}{\alpha}} \right)$ and define
$N_{\alpha}(\vX)$, $0 < \alpha \leq 2$, the entropy power of order
$\alpha$ as
\begin{equation}
  N_{\alpha}(\vX) = \frac{1}{e^{\frac{\alpha}{d} h({\bf \atilde{Z}})}} e^{\frac{\alpha}{d} h(\vX)}.
  \label{entrogen}
\end{equation}

\begin{theorem}[Generalized Isoperimetric Inequality for Entropies
  (GIIE)]
  \label{thGII}

  Let $\vX$ be a d-dimensional random vector such that both $h(\vX)$
  and $J_{\alpha}(\vX)$ exist, for some $1 < \alpha \leq 2$. Then
  \begin{equation}
    \frac{1}{d} \, N_{\alpha}(\vX)J_{\alpha}(\vX) \geq  \, \kappa_{\alpha} \eqdef  e^{(\alpha-1) \left(\psi(\alpha) +\gamma_{\text{e}}\right) - 1}, 
    \label{IPIgen}
  \end{equation}
  where $\gamma_{\text{e}}$ is the Euler-Mascheroni constant and
  $\psi(\cdot)$ is the digamma function.
\end{theorem}

Since $\psi(2) = -\gamma_{\text{e}} + 1$, we note right away that the
evaluation~(\ref{IPIgen}) for $\alpha =2$ yields the well known
IIE~\cite[Theorem 16]{cov-dem}:
\begin{equation*}
  \frac{1}{d} \, N(\vX) J(\vX) \geq 1,
\end{equation*}
with equality when $\vX$ is Gaussian distributed. For general values
of $1 < \alpha \leq 2$, whether equality in equation~(\ref{IPIgen}) is
achievable or not and under which conditions are still not answered.

\begin{proof}
  Let $t = \left( \frac{ \alpha \gamma^{\alpha}}{d} J_{\alpha} ( \vX )
  \right)^{\frac{1}{\alpha-1}}$ and $\vZ \sim \bm{\mathcal{S}} \left(
    \alpha, \gamma \right)$. By Theorem~\ref{bdsta},
  \begin{align}
    \frac{\alpha}{d}(h(\vX+\vZ) - h(\vX)) & = t^{\alpha-1}\,_2F_1\left(\alpha-1,\alpha-1;\alpha;-t\right) \nonumber \\
    & = \left(\frac{t}{1+t}\right)^{\alpha-1}\,_2F_1\left(\alpha-1,1;\alpha;\frac{t}{1+t}\right) \label{eqfor1},
  \end{align}
  where we used the fact that $t > 0$ and a transformation property of
  the Gauss hypergeometric function as presented in
  Appendix~\ref{asd}. Using the series representation of the Gauss
  hypergeometric function on the open unit disk, one can write:
  \begin{eqnarray*}
    _2F_1\left(\alpha-1,1;\alpha;\frac{t}{1+t}\right) &=& \sum^{+\infty}_{n=0} \frac{(\alpha-1)_n (1)_n}{(\alpha)_n} 
    \left(\frac{t}{1+t}\right)^{n}\\
    &=& \sum^{+\infty}_{n=0}\frac{\alpha -1}{n + \alpha -1} \left(\frac{t}{1+t}\right)^{n},
  \end{eqnarray*} 
  where $(A)_n = \frac{\Gamma(A+n)}{\Gamma(A)}$.
  Equation~(\ref{eqfor1}) can hence be written as,
  \begin{equation}
    \frac{\alpha}{d}(h(\vX+\vZ) - h(\vX)) =  (\alpha-1)\left(\frac{t}{1+t}\right)^{\alpha-1}\sum^{+\infty}_{n=0}\frac{1}{n 
      + \alpha -1} \left(\frac{t}{1+t}\right)^{n}. \label{finfor}
  \end{equation}

  The LHS of equation~(\ref{finfor}) is lower bounded by:
  \begin{equation*}
    \frac{\alpha}{d}(h(\vX+\vZ) - h(\vX)) \geq \frac{\alpha}{d}(h(\vZ) - h(\vX)) = \ln\frac{t^{\alpha-1}}{\frac{N_{\alpha}(\vX)
        J_{\alpha}(\vX)}{d}},
  \end{equation*}
  where we used equation~(\ref{entrogen}), the fact that $t =
  \left(\frac{\alpha\gamma^{\alpha}}{d}J_{\alpha}(\vX)\right)^{\frac{1}{\alpha-1}}$
  and that $h(\vZ) = d \ln \left(\gamma
    \alpha^{\frac{1}{\alpha}}\right) + h({\bf \atilde{Z}})$ in order
  to write the equality. As for the RHS of~(\ref{finfor}),
  \begin{align*}
    & (\alpha-1)\left(\frac{t}{1+t}\right)^{\alpha-1} \sum^{+\infty}_{n=0}\frac{1}{n + \alpha -1} \left(\frac{t}{1+t}\right)^{n} \\
    =  \, \, & (\alpha-1)\left(\frac{t}{1+t}\right)^{\alpha-1}\left[\frac{1}{\alpha-1} - \ln\left(1-\frac{t}{1+t}\right) - (\alpha-1)
      \sum^{+\infty}_{n=1}\frac{1}{n(\alpha+n-1)} \left(\frac{t}{1+t}\right)^n\right] \\
    = \, \, & \left(\frac{t}{1+t}\right)^{\alpha-1} + (\alpha-1)\left(\frac{t}{1+t}\right)^{\alpha-1}\ln(1+t)  -(\alpha-1)^{2} 
    \sum^{+\infty}_{n=1}\frac{1}{n(\alpha+n-1)} \left(\frac{t}{1+t}\right)^{n+\alpha-1}.
  \end{align*}
  
  Therefore~(\ref{finfor}) implies for any $t > 0$:
  \begin{multline*}
    \ln \left[ \frac{N_{\alpha}(\vX)J_{\alpha}(\vX)}{d}\right] -(\alpha-1) \ln t \geq -\left(\frac{t}{1+t}\right)^{\alpha-1} - (\alpha-1)
    \left(\frac{t}{1+t}\right)^{\alpha-1}\ln(1+t) \\
    +(\alpha-1)^{2} \sum^{+\infty}_{n=1}\frac{1}{n(\alpha+n-1)} \left(\frac{t}{1+t}\right)^{n+\alpha-1},
  \end{multline*}
  which by letting the scale $\gamma \rightarrow +\infty$ --and
  therefore $t \rightarrow +\infty$, gives the required result
  \begin{align}
    \ln N_{\alpha}(\vX)J_{\alpha}(\vX) &\geq (\alpha-1)^{2} \sum^{+\infty}_{n=1}\frac{1}{n(\alpha+n-1)} -1  \label{contser}\\ 
    &= (\alpha-1) \left(\psi(\alpha) +\gamma_{\text{e}}\right) - 1. \nonumber
  \end{align}

  The fact that the series $\sum_{n=1}^{+\infty}\frac{1}{n (\alpha +
    n-1)} \left( \frac{t}{1+t} \right)^{n+\alpha-1}$ is absolutely
  convergent permits the interchange in the order of the limits and
  justifies equation~(\ref{contser}).
\end{proof}


We plot in Figure~\ref{fig:figureuncer} the evaluation of the LHS of
equation~(\ref{IPIgen}) at the values of $\alpha = [1.2, 1.4, 1.6,
1.8]$ for alpha-stable RVs $\mathcal{S} \left( r, (r)^{-\frac{1}{r}}
\right)$ for the values of $r = [0.4, 0.6, 0.8, 1, 1.2, 1.4, 1.6,
1.8]$. The horizontal lines represent the RHS of
equation~(\ref{IPIgen}) for the considered values of $\alpha$. Note
that stable variables do not achieve the lower bound of the
GIIE~(\ref{IPIgen}) except when $\alpha =2$ where Gaussian variables
achieve the lower bound. The tightness in~(\ref{IPIgen})
  is explored in Figure~\ref{fig:figureuncerplus} where we evaluate
the product $N_{1.8}(X)J_{1.8}(X)$ whenever $X = X_1 + X_2$ where $X_1
\sim \mathcal{S} \left( r, (r)^{-\frac{1}{r}} \right)$ for $r= 1.8$
and $X_2 \sim \mathcal{N}(0,\sigma^2)$ for different value of
$\sigma$. The minimum is achieved for $\sigma = 4$ and not when $X$ is
alpha-stable (i.e., when $\sigma = 0$). Note that the computed minimum
in Figure~\ref{fig:figureuncerplus} is by no means a global minimum.

Whether there exist RVs that achieve the minimum of
$N_{\alpha}(X)J_{\alpha}(X)$ and whether the lower bound
$\kappa_{\alpha}$ is tight or not are still to be determined.

\begin{figure}[htb]
  \begin{center}
    \includegraphics[width=6in]{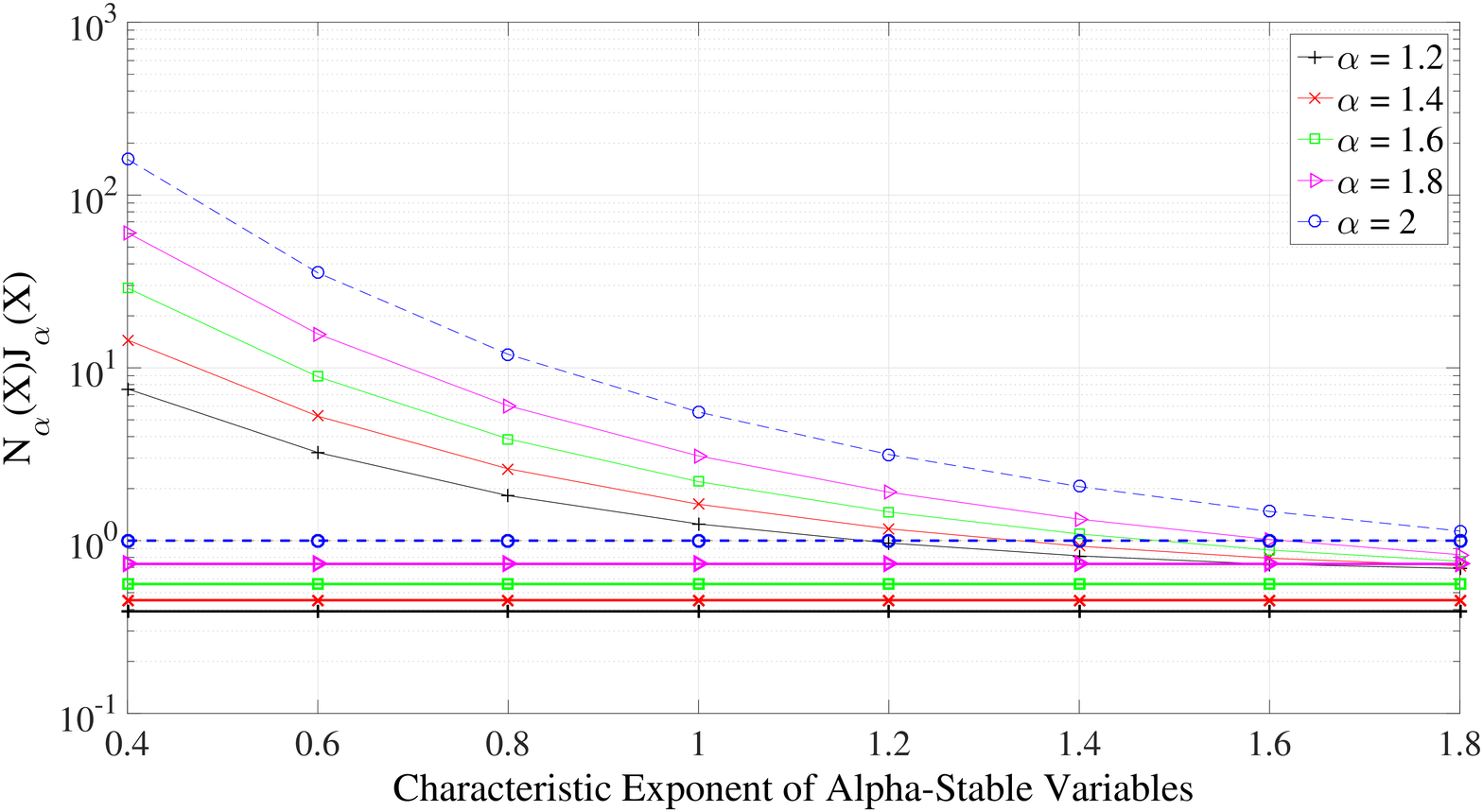}
    \caption{\small Evaluation of $N_{\alpha}(X)J_{\alpha}(X)$ and
      comparing it to $\kappa_{\alpha}$ for $X \sim
      \mathcal{S}\left(r,(r)^{-\frac{1}{r}}\right)$ for different
      values of $\alpha$ and $r$.
      \label{fig:figureuncer}}
  \end{center}
\end{figure}

\begin{figure}[htb]
  \begin{center}
    \includegraphics[width=5.5in]{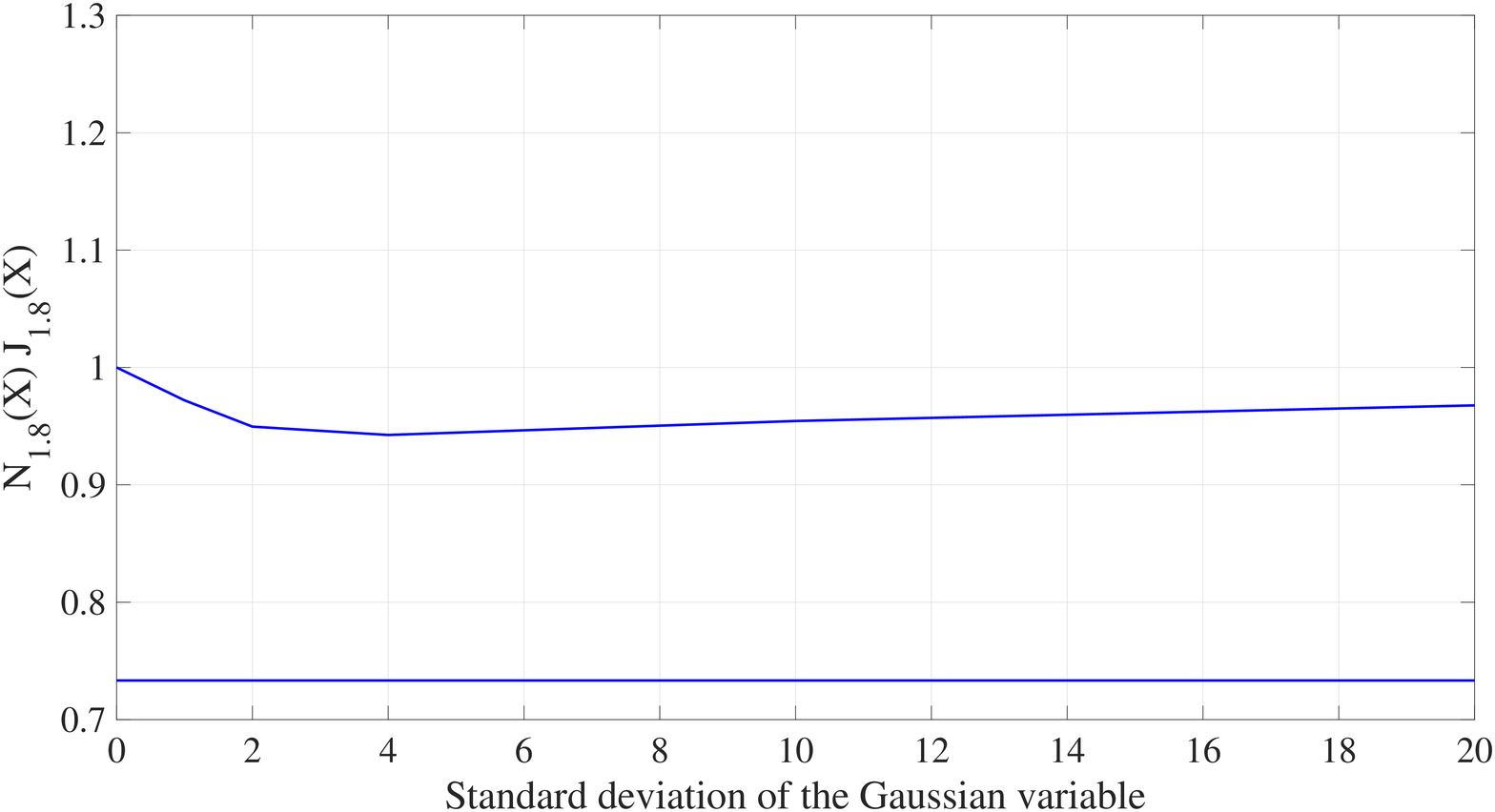}
    \caption{\small Comparison of $N_{1.8}(X)J_{1.8}(X)$ to
        $\kappa_{1.8} = 0.7333$ for $X = X_1 + X_2$ where $X_1 \sim
        \mathcal{S} \left( r, (r)^{-\frac{1}{r}} \right)$, $r = 1.8$
        and $X_2 \sim \mathcal{N}(0,\sigma^2)$ for different values of
        $\sigma$.
      \label{fig:figureuncerplus}}
  \end{center}
\end{figure}

Figure~\ref{fig:figuretight} shows the relative tightness of the lower
bound $\kappa_{\alpha}$ when the LHS of equation~(\ref{IPIgen}) is
evaluated at alpha-stable variables with characteristic exponents $r$
ranging from $0.4$ to $1.8$.  If we consider for example on the
$x$-axis the value of $r = 0.8$ which corresponds to the alpha-stable
variable $X \sim \mathcal{S}\left(r,(r)^{-\frac{1}{r}}\right)$, the
figure shows that as $\alpha$ decreases, $N_{\alpha}(X)J_{\alpha}(X)$
gets closer to $\kappa_{\alpha}$ in a relative manner.

\begin{figure}[h]
  \begin{center}
    \includegraphics[width=5.5in]{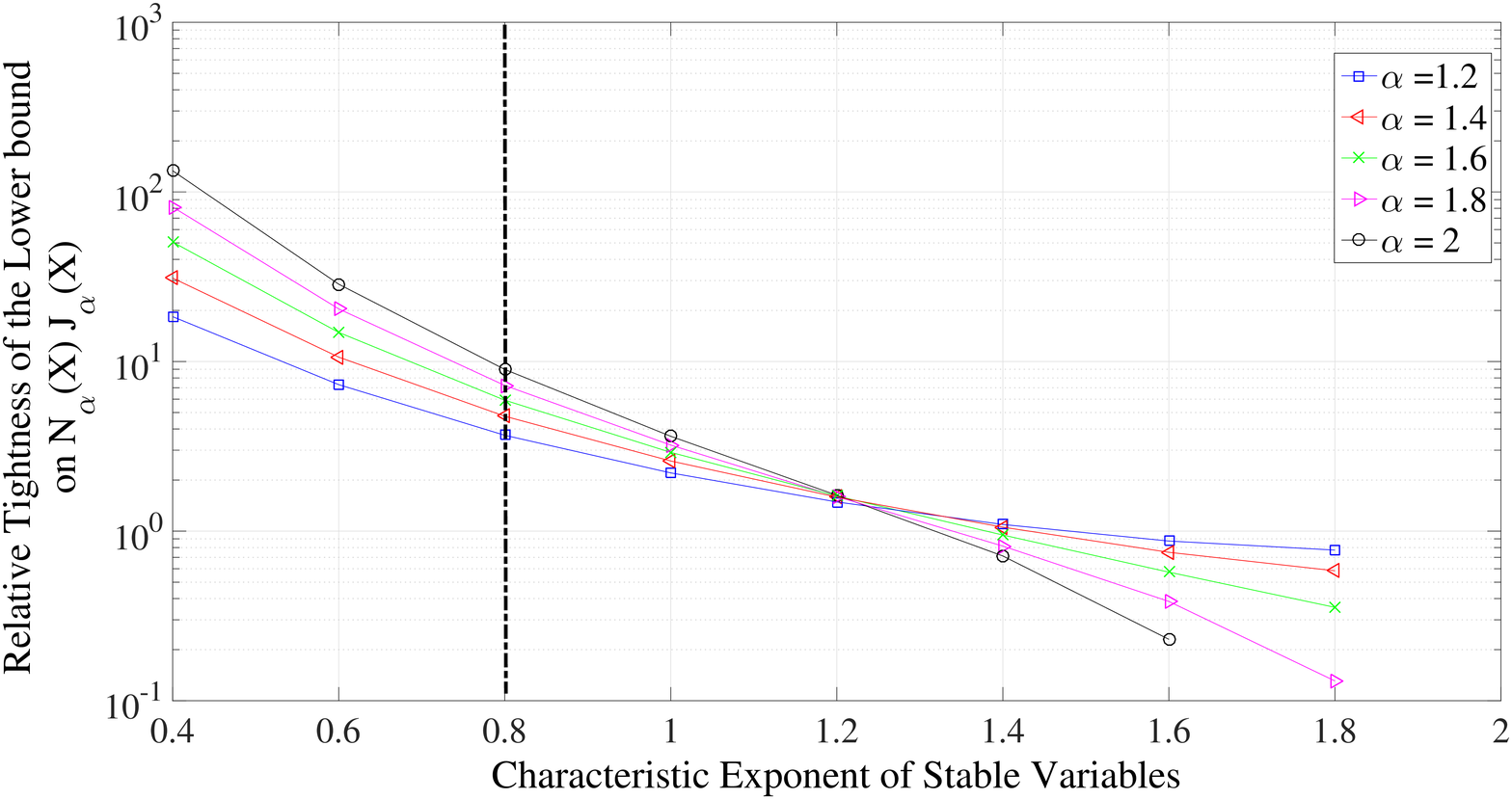}
    \caption{\small Relative tightness of $\kappa_{\alpha}$ for alpha-stable variables.
      \label{fig:figuretight}}
  \end{center}
\end{figure}


\section{Parameter Estimation in Impulsive Noise Environments: A
  Generalized Cramer-Rao Bound}
\label{sec:CRbd}

Consider now the problem of estimating a non-random vector of
parameters $\vtheta \in \Reals^{d}$ based on a noisy observation $\vX$
where the additive noise $\vN$ is of impulsive nature. Needless to say
that in this case the MSE criterion and the MMSE estimator are not
adequate.
More explicitly, let
\begin{equation*}
  \vX = \vtheta + \vN,
\end{equation*}
where $\vN$ is a noise variable having both $h(\vN)$ and
$J_{\alpha}(\vN)$ (for some $1 < \alpha \leq 2$) exist and finite.
Let $\hat{\vtheta}(\vX)$ be an estimator of $\vtheta$ based
on the observation of the random vector $\vX$.
%
%
%
A good indicator of the quality of the estimator
$\hat{\vtheta}(\vX)$ is the power of the ``error'' $\left(
  \hat{\vtheta}(\vX) - \vtheta \right)$. We find next a lower bound on
such metric which generalizes the previously known Cramer-Rao bound.
   
\begin{theorem}[Generalized Cramer-Rao Bound]
  \label{th:cramer-rao}
  
  Let $\hat{\vtheta}(\vX)$ be an estimator of the parameter $\vtheta$
  based on the observation $\vX = \vtheta + \vN$. Then the \apower
    of the error is lower bounded by
  \begin{equation}
    \label{CRfin}
    \aPw{\text{e}} = \aPw{\hat{\vtheta}(\vX) - \vtheta}  \geq \left( \frac{d \kappa_{\alpha}}{J_{\alpha}(\vN)} \right)^{\frac{1}{\alpha}}.
  \end{equation}
\end{theorem}
 
Note that whenever $\alpha = 2$ the result of
Theorem~\ref{th:cramer-rao} is the classical Cramer-Rao bound when
$\vN$ has IID components. It gives a looser version for general $\vN$:
\begin{equation}
  \label{gausscrb}
  \frac{\E{\left\|\hat{\vtheta}(\vX) - \vtheta\right\|^2}}{d} \geq \frac{d}{J(\vN)}.
\end{equation}

\begin{proof}
  Using Theorem~\ref{th:maxentro}, among all random vectors that
  have an \apower equal to $\aPw{\text{e}}$, the entropy maximizing
  variable is distributed according to $\bm {\mathcal{S}} \left(
    \alpha, \left( \frac{1}{\alpha} \right)^{\frac{1}{\alpha}}
    \aPw{\text{e}} \right)$ and
  \begin{equation*}
    h\left(\hat{\vtheta}(\vX) - \vtheta\right) \leq h({\bf \atilde{Z}}) + d \ln \aPw{\text{e}},
  \end{equation*}
  which implies that
  \begin{equation}
    \label{Nmaxprop}
    N_{\alpha}\left(\hat{\vtheta}(\vX) - \vtheta\right)  \leq \aPw{\text{e}}^{\alpha}.
  \end{equation}

  On the other hand, 
  \begin{equation}
      \label{Jprop}
      J_{\alpha}\left(\hat{\vtheta}(\vX) - \vtheta\right) = J_{\alpha}\left(\hat{\vtheta}(\vX) \right) 
      =J_{\alpha}\left(\hat{\vtheta}(\vX);\vtheta\right) \leq J_{\alpha}(\vX;\vtheta) = J_{\alpha}(\vN),
    \end{equation}
    where the second and the last equalities are due to
    Theorem~\ref{th:transprop} and the inequality is due to the data
    processing inequality for $J_{\alpha}(\cdot)$ proven in
    Theorem~\ref{dpif}. Applying the GIIE~(\ref{IPIgen}) to
  $\hat{\vtheta}(\vX) - \vtheta$, we obtain:
  \begin{equation*}
    \frac{N_{\alpha}\left(\hat{\vtheta}(\vX) - \vtheta\right) J_{\alpha}\left(\hat{\vtheta}(\vX) - \vtheta\right)}{d}  
    \geq \kappa_{\alpha},
  \end{equation*}
  which along with equations~(\ref{Nmaxprop}) and~(\ref{Jprop}) gives,
  \begin{equation*}
    \frac{J_{\alpha}(\vN) \,\aPw{\text{e}}^{\alpha}}{d}  \geq \kappa_{\alpha}.
  \end{equation*}
\end{proof}

Whenever the noise is a SS vector $\vN \sim \bm{ \mathcal{S} }(\alpha,
\gamma_N)$ for some $1 < \alpha \leq 2$, and since $J_{\alpha}(\vN) =
\frac{d}{\alpha \gamma_N^{\alpha}}$ by property (4),
Theorem~\ref{th:cramer-rao} specializes to the following bound.

\begin{corollary}[Generalized Cramer-Rao Bound for Stable Noise]
  
When the noise is a SS vector $\vN \sim \bm{ \mathcal{S}
    }(\alpha, \gamma_N)$, $1 < \alpha \leq 2$, the \apower
    $\aPw{\text{e}}$ of the error of all estimators $\hat{\vtheta}
    \left( \vX \right)$ is lower bounded by
  \begin{equation*}
    \aPw{\text{e}} \geq  \left(\alpha \kappa_{\alpha}\right)^{\frac{1}{\alpha}} \gamma_N.
  \end{equation*}
\end{corollary}


As an example, consider the Maximum Likelihood (ML) estimator
$\hat{\vtheta}_{\text{ML}}(\vX)$ which is given by,
\begin{equation*}
  \hat{\vtheta}_{\text{ML}}(\vX) = \arg\max_{\vtheta} \, \, \ln p_{\vN}(\vX-\vtheta).
\end{equation*}

Since $\vN$ is unimodal, $\hat{\vtheta}_{\text{ML}}(\vX) = \vX$, and
the \apower of the error $\hat{\vtheta}_{\text{ML}}(\vX) - \vtheta =
\vN$ is $\aPw{\text{e}} = \aPw{N} = (\alpha)^{\frac{1}{\alpha}}
\gamma_N = \left( \frac{d}{J_{\alpha}(\vN)}
\right)^{\frac{1}{\alpha}}$ for which equation~(\ref{CRfin}) holds
true.

\subsection*{The choice of $\alpha$:}
Note that equation~(\ref{CRfin}) establishes a new metric to measure
the average error strength and hence the estimator performance when
the noisy measurements are affected by an additive noise of impulsive
nature. The choice of a specific value of $\alpha$ is straightforward
whenever the noise belongs to the $\alpha$-parameterized domains of
normal attraction of stable variables. The quality of the estimator
$\hat{\vtheta}(\vX)$ is tied to the closeness of $\aPw{\text{e}}$ to
its lower bound, both of which are computable numerically as
previously shown for several probability laws. We mention that it is
not known in general whether equation~(\ref{CRfin}) is tight or
not. The tightness is already known when $\alpha=2$ for
$\hat{\vtheta}(\vX) = \vX$ and $\vN$ is a Gaussian vector. We believe
that answering the tightness question is equivalent to a similar
question about the GFII~(\ref{gfii}).

Finally, a direct implication of equation~(\ref{IPIgen}) is summarized
in the following: let $\aPw{\vX}$ denote the \apower of the random
vector $\vX$ according to equation~(\ref{newpow}). Using~(\ref{maxh}),
\begin{equation*}
  N_{\alpha}(\vX) \leq N_{\alpha}(\aPw{\vX} \atilde{\vZ}) = \frac{d}{J_\alpha( \aPw{\vX} \atilde{\vZ} )},
\end{equation*}
because
\begin{equation*}
  \frac{N_{\alpha}(\aPw{\vX} \atilde{\vZ}) J_{\alpha}(\aPw{\vX} \atilde{\vZ})}{d} = 1.
\end{equation*}

Equation~(\ref{IPIgen}) now yields,
\begin{equation*}
  J_\alpha(\vX) \geq \kappa_\alpha J_{\alpha}(\aPw{\vX} \atilde{\vZ})  = \kappa_\alpha \, \frac{d}{\aPw{\vX}^{\alpha}},
\end{equation*}
which is a generalization of the known fact that for any $\vX$
with covariance matrix of trace $d \sigma^2$, $J_{2}(\vX) \geq
J_{2}(\vZ) = \frac{d}{\sigma^2}$ where $\vZ \sim
\bm{\mathcal{N}}(\mu,\sigma^2 \mat{I})$ is a white Gaussian vector.

\section{Conclusion}
\label{con:Conclusion}

In a typical communication or measurement setup, the observed signal
is a noisy version of the signal of interest. Whether the source of
the noise comes from the equipment heating or an interferer, in many
instances, the effect of the perturbation is modeled in an additive
manner.  Generally, the role of a system designer is to build an
efficient system that recovers the information present in that
quantity of interest. In this work we highlighted various theoretical
aspects of such problems when the noise is heavy tailed, a scenario in
which alpha-stable distributions play a central role and find
applications in diverse fields of engineering and some other
disciplines.

Our main focus was on the parameter estimation problem in estimation
theory, where the basic estimation problem of the location parameter
of an alpha-stable variable is not yet well understood and performance
measures of a given estimator are to be further investigated. Since
the noise variable has an infinite second moment, standard tools such
as the second moment, the MSE and the Fisher information need to be
extended along with some inequalities satisfied by these information
measures.  Though the work of Gonzales~\cite{thgon} was in the
direction of some of these aspects, we believe that it is suitable for
the Cauchy case and not generic to the whole family of symmetric
stable distributions. Additionally, the work in~\cite{thgon} was with
a ``signal processing'' spirit.

We proposed in Section~\ref{seclocpow}, an expression to evaluate the
power of signals in symmetric alpha-stable noise environments. Though
the definition of the \apower has unfamiliar format where the value of
the power is incorporated within a cost function, it depends on an
average of a logarithmically tailed cost function. Besides the
logarithmic tail behavior of the averaged function, the main argument
for suggesting $\aPw{X}$ as defined in Definition~\ref{powdef}, is to
find a definition that is generic for the stable space of noise
distributions, including the Gaussian since stable distributions are
the most common noise models encountered by virtue of the generalized
CLT. Definition~\ref{powdef} is chosen to become the standard
deviation in the Gaussian case in order to unify the order of the
power operator in such a way if the variable is linearly scaled then
the power also scales linearly.  We proved that
Definition~\ref{powdef} defines a space where the alpha-stable noise
is the worst in terms of entropy/randomness which implies that the
alpha-stable channel model is a worst-case scenario whenever there is
an impulsive noise assumption.  This fact mimics the role of the
Gaussian variable among the finite variance space of RVs and
generalize it to an equivalent role of stable variables among the
space of RVs that have a finite power $\aPw{X}$.

A generalized notion of the Fisher information is introduced in
Section~\ref{sc:GFI} and is shown to satisfy standard information
measures properties: positiveness, scalability, additivity, etc. The
newly defined quantity $J_{\alpha}(\cdot)$ is shown to abide by
fundamental identities and relationships such as a chain rule, a
generalized Fisher information inequality and a generalized
isoperimetric inequality for entropies. These lead to a generalized
Cramer-Rao bound proven in Section~\ref{sec:CRbd} which sets a novel
lower bound on the \apower of the estimation error for any estimator
of a location parameter. This bound can be used to characterize the
performance of estimators in impulsive noise environments and
naturally opens the door to the related problems of efficiency and
optimality of estimators.

The newly defined power measure $\aPw{X}$ establishes a novel way to
approach communication theoretic problems. As an example, the
classical approach to the channel capacity problem is done from a
channel input perspective. Under this perspective and for the purpose
of emulating real scenarios, input signals are supposed to abide by
some power constraints such as the second moment constraint. Assuming
that the additive noise would also have a finite second moment, this
approach quantified the different metrics of the channel with respect
to the input power measure irrespective of the noise model. As an
example, the capacity of the linear additive Gaussian channel under an
average power constraint is given by the famous formula ``$C=
\frac{1}{2}\ln (1 + \text{SNR})$'' where the ``$\text{SNR}$'' is the
signal to noise ratio between the variance of the input to that of the
Gaussian noise, hence relating the input power as defined for the
input space to the noise power since the noise falls within the input
space.  Naturally, this approach breaks when the noise is not of the
same ``nature'' as the input space. This is true for impulsive noise
models such as the alpha-stable ones having infinite second moments
which do not belong to the input space of finite power (second moment)
RVs. Since the performance of any adopted strategy at the input is
viewed by its effect at the output end, it seems reasonable to
consider the additive channel while imposing a ``quality'' constraint
on the output. By restricting the output space to satisfy certain
power requirements, we are indirectly taking into consideration the
nature of the noise in the formulation of the constraint which
constructs an input space of variables of the same ``nature'' of the
noise. This is in accordance with the fact that the system designer
has no control over the noise model which is dictated by the channel
and can assume the possibility of choosing from an input space
similar in nature to that of the noise, the input signal that best
overcomes the noise effect. For the linear AWGN channel, exceptionally
the output approach gives exactly the same answer as the input
approach: constraining the output average power implies a constraint
on the input average power.

Finally, we emphasize that the generalized tools and identities
presented in this work constitutes an ``extension'' of the Gaussian
estimation theory to a stable estimation theory in general and may be
viewed as complementary to the works found in the literature by
answering some ``fundamental-limits'' questions.



\appendices
\section{Multivariate Alpha-Stable Distributions, Riesz Potentials and
  Hypergeometric Functions}
\label{asd}

\subsection{Univariate Alpha-Stable Distributions}
\label{appmult}

\begin{definition}[Univariate Stable Distributions]
  
  A univariate stable RV $X \sim \mathcal{S}(\alpha, \beta, \gamma,
  \delta)$ is one with characteristic function,
  \begin{align*}
    &\phi_{X}(\omega) = \exp\left[i \delta \omega -\gamma^{\alpha} \left(1 - i \beta \sgn(\omega) 
        \Phi(\omega) \right) |\omega|^{\alpha}\right]\\
    &\hspace{-0.25cm}\biggl(0< \alpha \leq 2 \quad -1 \leq \beta \leq 1 \quad \gamma > 0 \quad \delta \in \Reals \biggr),
  \end{align*}
  where $\sgn(\omega)$ is the sign of $\omega$ and $\Phi(\cdot)$ is
  given by:
  \begin{equation*}
    \Phi(\omega) = \left\{ \begin{array}{ll} 
        \displaystyle \tan \left(\frac{\pi \alpha}{2} \right) \quad & \alpha \neq 1 \\
        \displaystyle  -\frac{2}{\pi} \ln|\omega| \quad & \alpha = 1.
      \end{array} \right.
  \end{equation*} 
  The constant $\alpha$ is called the ``characteristic exponent'',
  $\beta$ is the ``skewness'' parameter, $\gamma$ is the ``scale''
  parameter ($\gamma^{\alpha}$ is often called the ``dispersion'') and
  $\delta$ is the ``location'' parameter.
\end{definition}

We make the following specifications:
\begin{itemize}
\item Whenever the parameters $\beta = 0$ and $\delta=0$, the stable
  variable is symmetric and denoted $X \sim
  \mathcal{S}(\alpha,\gamma)$.

\item The case where $\alpha = 2$ corresponds to the Gaussian RV $X
  \sim \mathcal{S}(2,0,\gamma,\delta) =
  \mathcal{N}(\delta,2\gamma^2)$.

\item Whenever $|\beta| = 1$, the alpha-stable variable is called
  totally-skewed. Furthermore, it is one sided when $\alpha < 1$.
\end{itemize}

\subsection{Multivariate Alpha-Stable Distributions}

\begin{definition}[Sub-Gaussian Symmetric Alpha-Stable]
  \cite[p.78 Definition 2.5.1]{sam1994} 
  \label{def:subgauss}

  Let $0 < \alpha <2$ and let $\mathcal{A} \sim \mathcal{S} \left(
    \frac{\alpha}{2}, 1, \left( \text{cos} \left( \frac{\pi \alpha}{4}
      \right) \right)^{\frac{2}{\alpha}}, 0 \right)$ be a totally
  skewed one sided alpha-stable distribution. Define $\vG =
  (G_1,\cdots,G_d)$ to be a zero mean Gaussian vector in
  $\Reals^{d}$. Then the random vector $\vN = (A^{\frac{1}{2}}G_1,
  \cdots , A^{\frac{1}{2}}G_d)$ is called a sub-Gaussian symmetric
  alpha-stable (\SaS) random vector in $\Reals^{d}$ with underlying
  vector $\vG$. In particular, each component $A^{\frac{1}{2}}G_i$, $1
  \leq i \leq d$ is a \SaS variable with characteristic exponent
  $\alpha$.  In this work we only use sub-Gaussian \SaS vectors such
  that the underlying Gaussian vector has IID zero-mean components
  with variance $2\gamma^2$, for some $\gamma > 0$.  We denote such a
  vector as $\bm{\mathcal{S}} \left( \alpha, \gamma \right)$.
\end{definition}  

\begin{proposition}~\cite[p.79 Proposition 2.5.5]{sam1994}
  \label{thdef0}
  
  Let $\vN = (N_1, \cdots , N_d)$ be a sub-Gaussian \SaS with an
  underlying Gaussian vector having IID zero-mean components with
  variance $2 \gamma^2$, for some $\gamma > 0$. Then, the
  characteristic function of $\vN$ is:
  \begin{equation*}
    \phi_{\vN} (\vomega)  = e^{-\gamma^{\alpha} \|\vomega\|^{\alpha}}.
  \end{equation*}
  The RVs $N_i$s, $1 \leq i \leq d$, are dependent and each
  distributed according to $\mathcal{S}(\alpha,\gamma)$.
\end{proposition}


\begin{property}[Isotropic property]
  \label{prop:tail}

  Let $\vN = (N_1, \cdots , N_d) \sim \bm{\mathcal{S}} (\alpha,
  \gamma)$. Then, for ${\bf n} \neq {\bf 0}$
  \begin{equation}
    p_{\vN}({\bf n}) = f(\|{\bf n}\|) = \frac{\Gamma\left(\frac{d}{2}\right)}{2\pi^{\frac{d}{2}}}\|{\bf n}\|^{1-d}
    p_{R}\left(\|{\bf n}\|\right),
    \label{eq:rotation}
  \end{equation}
  where $R = \|\vN\|$ is the amplitude of $\vN$ and $p_{R}(\cdot)$ is
  its density function. Furthermore, we have:
  \begin{equation*}
    \lim_{r \rightarrow +\infty} r^{1+\alpha}p_{R}(r) =  \alpha \gamma^{\alpha}  k_1, \qquad 
    k_1 = 2^{\alpha}\frac{\sin \left(\frac{\pi \alpha}{2}\right)}{\frac{\pi \alpha}{2}}\frac{\Gamma\left(\frac{2+ \alpha}{2}\right) 
      \Gamma\left(\frac{d+ \alpha}{2}\right)}{\Gamma\left(\frac{d}{2}\right)}.
  \end{equation*}
\end{property}

\begin{proof}
  Refer to~\cite{Nolan2013}.
\end{proof}

Note that by equation~(\ref{eq:rotation}), $\bm{\mathcal{S}} \left(
  \alpha, \gamma \right)$ is isotropic.

\begin{property}
  \label{prop:bd}

  Let $\vN = (N_1, \cdots , N_d) \sim \bm{\mathcal{S}} \left( \alpha,
    \gamma \right)$. Then $\frac{\partial p_{\vN}}{\partial n_i} \,
  \frac{1}{p_{\vN}} \left ({\bf n}\right)$ is bounded for all $1 \leq
  i \leq d$.
\end{property}

\begin{proof}
  Since $p_{\vN}(\cdot)$ is infinitely differentiable, it is enough to
  show boundedness at large values of $r = \|\vn\|$. We use the
  results of Property~\ref{prop:tail} to write
  \begin{equation}
    \label{eq:tailvec}
    \lim_{r \rightarrow \infty} \frac{p_{\vN}({\bf n})}{\|{\bf n}\|^{-d-\alpha}} = \lim_{r \rightarrow \infty} \frac{\Gamma\left(\frac{d}{2}
      \right)}{2\pi^{\frac{d}{2}}} \frac{\|{\bf n}\|^{1-d}p_{R}\left(\|{\bf n}\|\right)}{\|{\bf n}\|^{-d-\alpha}} =  
    \lim_{r \rightarrow \infty} \frac{\Gamma\left(\frac{d}{2}\right)}{2\pi^{\frac{d}{2}}} r^{1 + \alpha} p_R(r) = \frac{\alpha 
      \gamma^{\alpha} k_1 \Gamma\left(\frac{d}{2}\right)}{2\pi^{\frac{d}{2}}} = k_2,
  \end{equation}
  where $k_1$ is defined in Property~\ref{prop:tail}. Using
  l'H\^{o}pital's rule, we write
  \begin{eqnarray}
    k_2 &=& \lim _{r \rightarrow \infty} \frac{\frac{d}{dr}p_{\vN}({\bf n})}{\frac{d}{dr}r^{-d-\alpha}} \nonumber \\
    &=&  -\frac{1}{d + \alpha}  \lim _{r \rightarrow \infty} \frac{\frac{\partial}{\partial n_1}p_{\vN}({\bf n}) \times 
      \frac{\partial n_1}{\partial r} + \cdots + \frac{\partial }{\partial n_d}p_{\vN}({\bf n}) \times \frac{\partial n_d}
      {\partial r}}{r^{-d-\alpha-1}}\nonumber \\
    &=& -\frac{1}{d + \alpha}  \lim _{r \rightarrow \infty} \frac{r\left(\frac{1}{n_1}\frac{\partial}{\partial n_1}p_{\vN}({\bf n})
        + \cdots + \frac{1}{n_d}\frac{\partial}{\partial n_d}p_{\vN}({\bf n})\right)}{r^{-d-\alpha-1}}\nonumber\\
    &=& -\frac{1}{d + \alpha} \lim _{r \rightarrow \infty} \frac{\frac{1}{n_1}\frac{\partial}{\partial n_1}p_{\vN}({\bf n})
      + \cdots + \frac{1}{n_d}\frac{\partial}{\partial n_d}p_{\vN}({\bf n})}{r^{-d-\alpha-2}} \label{eq:lim}
  \end{eqnarray}

  Using equation~(\ref{eq:rotation}), $p_{\vN}({\bf n})$ is
    decreasing in $r$~\cite[Section 2.1]{Nolan2013}. Therefore,
  $\frac{1}{n_i}\frac{\partial}{\partial n_i} p_{\vN} ({\bf n})$ is
  negative for all $1 \leq i \leq d$ and $-\frac{1}{n_i}
  \frac{\partial}{\partial n_i} p_{\vN} ({\bf n}) \leq
  -\left(\frac{1}{n_1} \frac{\partial}{\partial n_1} p_{\vN} ({\bf n})
    + \cdots + \frac{1}{n_d} \frac{\partial}{\partial n_d} p_{\vN}
    ({\bf n}) \right)$. Hence, equation~(\ref{eq:lim}) implies
  \begin{eqnarray*}
    k_2 &=& -\frac{1}{d + \alpha} \lim _{r \rightarrow \infty} \frac{\frac{1}{n_1}\frac{\partial}{\partial n_1}p_{\vN}({\bf n})
      + \cdots + \frac{1}{n_d}\frac{\partial}{\partial n_d}p_{\vN}({\bf n})}{r^{-d-\alpha-2}}\\
    &\geq&\frac{1}{d + \alpha} \lim _{r \rightarrow \infty} \frac{-\frac{1}{n_i}\frac{\partial}{\partial n_i}p_{\vN}({\bf n})}
    {r^{-d-\alpha-2}},
  \end{eqnarray*}
  which implies that there exists a constant $\kappa_{i}$ such
  that at large values of $r$, $-\frac{1}{n_i}
  \frac{\partial}{\partial n_i} p_{\vN} ({\bf n}) \leq \kappa_{i} \,
  r^{-d-\alpha-2}$ for $1 \leq i \leq d$. The fact that $\lim_{r
    \rightarrow +\infty} r^{d+ \alpha} p_{\vN} ({\bf n}) = k_2$
  completes the proof.
\end{proof}

\subsection{Riesz Potentials}
\label{apprie}

\begin{definition}[Riesz Potentials]~\cite[p.117 Section 1]{stein70}

  Let $ 0 < \nu < 1$. The Riesz potential $I_{\nu}(f)(x)$ for a
  sufficiently smooth $f: \Reals^d \rightarrow \Reals$ having a
  sufficient decay at $\infty$ is given by:
  \begin{equation*}
    I_{\nu}(f)(\vx) = \frac{1}{\kappa(\nu)} \int_{\Reals^d} \|\vx-\vy\|^{-d+\nu}f(\vy)\,d\vy, \qquad
    \kappa(\nu) = \pi^{\frac{d}{2}}2^{\nu} \frac{\Gamma \left(
        \frac{\nu}{2} \right)}{\Gamma \left( \frac{d}{2} - \frac{\nu}{2}
      \right)}.
  \end{equation*}
\end{definition}

\begin{property}
  \label{propro}

  Among other properties, $\text{I}_{\nu}(f)$ satisfies the following:
  \begin{itemize}
  \item $\mathcal{F}\left(\text{I}_{\nu}(f)\right)(\vomega) =
    \|\vomega\|^{-\nu} \mathcal{F}(f(\vx))(\vomega)$ in the
    distributional sense.
  \item $\text{I}_{0}(f)(\vx) \eqdef \lim_{\nu \rightarrow 0}
    \text{I}_{\nu}(f)(\vx)= f(\vx)$.
  \item Whenever $\int |f|(\vx) \text{I}_{\nu}(|g|)(\vx) \, d\vx$ is
    finite, we have:
    \begin{equation*}
      \int f(\vx) \text{I}_{\nu}(g)(\vx) \, d\vx= \int
      \text{I}_{\nu}(f)(\vx) g(\vx)\,d\vx.
    \end{equation*}
  \end{itemize}
\end{property}

\subsection{Hypergeometric Functions}

\begin{definition}[Gauss Hypergeometric Functions]

  For generic parameters $a, b, c$, the Gauss hypergeometric
  function $_2F_1(a,b;c;z)$ is defined as the following power series:
  \begin{equation*}
    _2F_1(a,b;c;z) = \sum_{i = 0}^{+\infty} \frac{(a)_i(b)_i}{(c)_i i!}z^i, \qquad |z| < 1.
  \end{equation*}
  Outside of the unit circle $|z| < 1$, the function is defined as the
  analytic continuation of this sum with respect to $z$, with the
  parameters $a$, $b$ and $c$ held fixed. The notation $(d)_i$ is
  defined as:
  \begin{equation*}
    (d)_i = \left\{\begin{array}{ll} 
        \displaystyle 1 \quad & i = 0\\
        \displaystyle d(d+1) \dots (d+i-1) \quad & i > 0.
      \end{array} \right.
  \end{equation*}
\end{definition}

\begin{proposition}
  The Gauss hypergeometric function $_2F_1(a,b;c;z)$ satisfies the
  following property:
  \begin{equation*}
    _2F_1(a,b;c;z) = (1-z)^{-a} _2F_1\left(a,c-b;c;\frac{z}{z-1}\right), \,\,\,\,\,\,\, z \notin (1,+\infty).
  \end{equation*}
  
\end{proposition}

\section{Properties of $\aPw{\vX}$}
\label{ap:prop_aPw}

We consider random vectors $\vX$ and $\vY$ such that:
\begin{equation*} 
  \left\{ \begin{array}{ll}
      \E{\ln\left(1 + \| \vX \|\right)} < \infty, & \text{ when considering cases where } \alpha < 2. \\
      \E{ \| \vX \|^2} < \infty & \text{ when considering cases where } \alpha = 2.
    \end{array} \right.
\end{equation*} 

We first start by establishing the following Lemmas:
\begin{lemma}
  \label{lm:aPw_1}
  Let $\vX \neq {\bf 0}$ and define the function of $P > 0$,
  \begin{equation*}
    g\left(\text{P}\right) \eqdef - \E { \ln p_{{\bf \atilde{Z}}} \left( \frac{\vX}{\text{P}} \right) }.
  \end{equation*}
  The function $g\left(\text{P}\right)$ is continuous and decreasing
  on $\Reals^{+} \setminus \{0\}$.
\end{lemma} 

\begin{proof}  \mbox{}
   
  \underline{Continuity:} Let $\text{P}_0 > 0$, then
  \begin{eqnarray*}
    -\lim_{\text{P} \rightarrow \text{P}_0} \E{ \ln p_{{\bf \atilde{Z}}} \left(\frac{\vX}{\text{P}} \right) }
    &=& -\lim_{\text{P} \rightarrow \text{P}_0}\int \ln p_{{\bf \atilde{Z}}}\left(\frac{\vx}{\text{P}}\right)\,dF(\vx)\\
    & = &- \int \lim_{\text{P} \rightarrow \text{P}_0}  \ln p_{{\bf \atilde{Z}}}\left(\frac{\vx}{\text{P}}\right)\,dF(\vx)\\
    &= & - \int \ln p_{{\bf \atilde{Z}}}\left(\frac{\vx}{\text{P}_{0}}\right)\,dF(\vx),
  \end{eqnarray*} 
  where in order to write the last equation we used the fact that
  $p_{{\bf \atilde{z}}}(\cdot)$ is continuous on $\Reals^d$. The
  interchange in the order between the limit and the integral signs is
  justified using DCT as follows: In a neighbourhood of $\text{P}_0$,
  choose a $\tilde{\text{P}}$ such that $0< \tilde{\text{P}} <
  \text{P}$. Since $p_{{\bf \atilde{Z}}}(\vx)$ is rotationally
  symmetric and decreasing in $\| \vx \|$,
  \begin{equation*}
    \left|\ln p_{{\bf \atilde{Z}}}\left(\frac{\vx}{\text{P}}\right)\right| \leq  
    \left|\ln p_{{\bf \atilde{Z}}}\left(\frac{\vx}{\tilde{\text{P}}}\right)\right|, \qquad \forall \, \vx \in \Reals^d,
  \end{equation*}
  which equality only at $\vx = {\bf 0}$. Therefore,
  \begin{equation*}
    \E{\left| \ln p_{{\bf \atilde{Z}}} \left(\frac{\vX}{\text{P}} \right) \right| }
    \leq \E{\left| \ln p_{{\bf \atilde{Z}}} \left( \frac{\vX}{\tilde{\text{P}}} \right) \right|},
  \end{equation*} 
  which is finite because $\E{\ln\left(1 + \| \vX \|\right)} < \infty$
  whenever ${\bf \atilde{Z}}$ is sub-Gaussian by virtue of the fact
  that $ p_{{\bf \atilde{Z}}} ( \vx ) = \Theta \left( \frac{1}{\| \vx
      \|^{d+\alpha}} \right)$\footnote[9]{In this work, we say that
    $f(x) = \Omega \left(g(x)\right)$ if and only if $\exists \,
    \kappa > 0, c > 0$ such that $|f(x)| \geq \kappa |g(x)|, \forall
    |x| \geq c$.  Equivalently, we say that $g(x) =
    O\left(f(x)\right)$. We say that $f(x) = \Theta \left(g(x)\right)$
    if and only if $f(x) = O \left(g(x)\right)$ and $f(x) = \Omega
    \left(g(x)\right)$.}
  (see Appendix~\ref{asd}) and because it is assumed that $\E{\|{\bf
      X}\|^2} < \infty$ whenever $\alpha = 2$ and ${\bf \twotilde{Z}}$
  is Gaussian.

  \underline{Monotonicity:} Let $0< \tilde{\text{P}} <
  \text{P}$. Since $p_{{\bf \atilde{Z}}}(\vx)$ is rotationally
  symmetric and decreasing in $\|\vx\|$, $p_{{\bf \atilde{Z}}}(\vx /
  \text{P}) \geq p_{{\bf \atilde{Z}}}(\vx / \tilde{\text{P}})$ for all
  $\vx$, with equality only at ${\bf 0}$. Since $\vX \neq {\bf 0}$,
  there exists a non-zero point of increase\footnote{A vector $\vx$ is
    said to be a point of increase if and only if, $\Pr(\|\vX - \vx\|
    < \eta)>0$ for all $\eta>0$.} $\vx_o$, and $g \left( \text{P}
  \right)$ is decreasing in $\text{P} > 0$.
\end{proof}

We evaluate next the limit values of $g(\text{P})$ at $0$ and
$+\infty$.

\begin{lemma}
  \label{lm:aPw_2}
  In the limit, 
  \begin{equation*}
    \lim_{\text{P} \rightarrow 0}  g\left(\text{P}\right) = +\infty \quad \& \quad 
    \lim_{\text{P} \rightarrow +\infty}g(\text{P}) < h\left({\bf \atilde{Z}}\right).
  \end{equation*}
\end{lemma}

\begin{proof} \mbox{}
  
  \underline{The limit at zero:} Since $\vX \neq {\bf 0}$, there exits
  a $\delta > 0$ such that $\text{Pr}\left(\| \vX \| \geq
    \delta\right) > 0$ and
  \begin{eqnarray*}
    g\left(\text{P}\right) &=&- \int_{\| \vx \| \leq \delta}  \ln p_{{\bf \atilde{Z}}}\left(\frac{\vx}{\text{P}}\right)
    \,dF_{\vX}(\vx) - \int_{\| \vx \| \geq \delta}  \ln p_{{\bf \atilde{Z}}}\left(\frac{\vx}{\text{P}}\right)
    \,dF_{\vX}(\vx) \\
    &\geq& -\text{Pr}\left(\|X\| \leq \delta\right)\ln p_{{\bf \atilde{Z}}}(0) - \text{Pr}\left(\|X\| \geq \delta\right)  
    \ln p_{{\bf \atilde{Z}}}\left(\frac{\delta}{\text{P}}\right),
  \end{eqnarray*}
  because $p_{{\bf \atilde{Z}}}(\vx)$ is decreasing in $\| \vx \|$.
  Since $p_{{\bf \atilde{Z}}}(\vx) \rightarrow 0$ as $\| \vx \|
  \rightarrow +\infty$, then $\displaystyle \lim_{\text{P} \rightarrow
    0} g\left(\text{P}\right) = +\infty$.

  \underline{The limit at infinity:} Computing the limit at infinity,
  \begin{eqnarray*}
    \lim_{\text{P} \rightarrow +\infty}g(\text{P}) &=& \lim_{\text{P} \rightarrow +\infty} - \int_{\Reals^d} \ln p_{{\bf \atilde{Z}}}
    \left(\frac{\vx}{\text{P}}\right)\,dF_{\vX}(\vx)\\
    &=& - \int_{\Reals^d} \lim_{\text{P} \rightarrow +\infty}  \ln p_{{\bf \atilde{Z}}}\left(\frac{\vx}{\text{P}}\right)\,dF_{\vX}(\vx)\\
    &=& -\ln p_{{\bf \atilde{Z}}}(0) < h\left({\bf \atilde{Z}}\right),
  \end{eqnarray*}
  where the last inequality is true because $p_{{\bf
      \atilde{Z}}}(\vx)$ is decreasing in $\| \vx \|$. The interchange
  between the limit and the integral sign is due to DCT as shown in
  the proof of Lemma~\ref{lm:aPw_1}.

\end{proof}

\begin{lemma}
  \label{lem:loc}
  Let $\vX$ be a random vector that has a rotationally symmetric PDF
  that is non-increasing in $\| \vx \|$, then
  \begin{equation*}
    - \E{ \ln p_{\atilde{\vZ}}(\vX - \bm{\nu}) } \geq - \E{ \ln p_{\atilde{\vZ}}(\vX ) }, 
    \qquad \forall \, \bm{\nu} \in \Reals^d.
  \end{equation*}
\end{lemma}

%
%

\begin{proof}
  Since $p_{\vX}(\cdot)$ and $ p_{{\bf \atilde{Z}}}(\cdot)$ are
  rotationally invariant, one can restrict the proof to the case where
  all the $\{\nu_i\}_{1 \leq i \leq d}$'s are non-positive by applying
  an appropriate rotation transformation to the variable of
  integration. Hence, for $\{\nu_i\}_{1 \leq i \leq d}$ non-positive,
  taking the partial derivative of
  \[
  - \E{ \ln p_{\atilde{\vZ}}(\vX - \bm{\nu}) } = -\int_{\Reals^d} p_{\vX}(\vx) \ln p_{{\bf \atilde{Z}}} \left( 
    \vx - \vnu \right)\, d \vx,
  \]
  with respect to $\nu_i$ and interchanging the integral and the
  derivative yields
  \begin{align*}
    \int^{+\infty}_{-\infty} p_{\vX}(\vx) \, \frac{\partial}{\partial z_i}p_{{\bf \atilde{Z}}}
    (\vx - \vnu) \frac{1}{p_{{\bf \atilde{Z}}}(\vx - \vnu)} \,d\vx
    =   \E{ \frac{\partial}{\partial z_i}p_{{\bf \atilde{Z}}} (\vX - \vnu ) \frac{1}{p_{{\bf \atilde{Z}}}(\vX - \vnu ) }} \geq 0,
  \end{align*}
  which is true by virtue of the facts that $p_{\vX}(\vx)$ is
  rotationally symmetric, non-increasing in $\| \vx \|$, that for all
  $1 \leq i \leq d$ the derivative function $\frac{\partial}{\partial
    z_i} p_{{\bf \atilde{Z}}} \frac{1}{p_{{\bf \atilde{Z}}}} (\vx)$ is
  an odd function that is non-positive on $x_i \geq 0$ and that
  $\{\nu_i\}_{1 \leq i \leq d}$ are non-positive. This implies that
  $-\E{ \ln p_{{\bf \atilde{Z}}} ( \vX - \vnu ) }$ is maximum at $\vnu
  = {\bf 0}$.

  The interchange between the derivative and the expectation operator
  is justified by Lebesgue's DCT since the integrand $\left[
    \frac{\partial p_{{\bf \atilde{Z}}}}{\partial z_{i}} (\vx)
    \frac{1}{p_{\bf \atilde{Z}} (\vx)} \right]$ is bounded by
  Property~\ref{prop:bd} in Appendix~\ref{asd}.
\end{proof}

We prove in what follows some properties of the \apower set in
Definition~\ref{powdef}.

\begin{itemize}
\item[(i)] $\aPw{\vX}$ exists, is unique and satisfies property R1,
  i.e. $\aPw{\vX} \geq 0$ with equality if and only if $\vX = {\bf
    0}$.
  Indeed, for a non-zero random vector, using the continuity of
  $g(\text{P})$ and the fact that it is decreasing from $+\infty$ to
  $-\ln p_{{\bf \atilde{Z}}}({\bf 0\/}) < h({\bf \atilde{Z}})$ proven
  in Lemmas~\ref{lm:aPw_1} and~\ref{lm:aPw_2}, there {\em exists\/} a
  {\em positive\/} and {\em unique\/} $\aPw{\vX}$ such that
  equation~(\ref{newpow}) is satisfied which proves property (i).

\item[(ii)] $\aPw{\vX}$ satisfies property R2. In fact, for any
  $a \in \Reals$,
  \begin{equation*}
    \aPw{a \vX} = |a| \, \aPw{\vX}.
  \end{equation*}
  This directly follows from equation~(\ref{newpow}) and the fact that
  $p_{{\bf \atilde{Z}}}(\cdot)$ is rotationally symmetric.

\item[(iii)] Let $\vX$ and $\vY$ be two independent random vectors and
  assume that $\vY$ has a rotationally symmetric PDF that is
  non-increasing in $\|\vy\|$. Let $\vZ = \vX + \vY$, then
  $\aPw{\vZ} \geq \aPw{\vY}$. Indeed,
  \begin{eqnarray}
    -\E{ \ln p_{{\bf \atilde{Z}}} \left( \frac{\vZ}{\aPw{\vY}}\right) } &=&  \Ep{\vX}
    { - \Ep{\vY}{ \ln p_{{\bf \atilde{Z}}} \left( \frac{\vx+\vY}{\aPw{\vY}}\right)
        \Big| \vX } } \nonumber\\
    &\geq& \Ep{\vX}{ -\Ep{\vY}{ \ln p_{{\bf \atilde{Z}}} \left( \frac{\vY}{\aPw{\vY}}\right) \Big|\vX } }\label{minloc}\\
    &=&-\Ep{\vY}{\ln p_{{\bf \atilde{Z}}} \left( \frac{\vY}{\aPw{\vY}} \right) }\nonumber\\
    &=& h({\bf \atilde{Z}}), \label{proppow}
  \end{eqnarray}
  where equation~(\ref{minloc}) is an application of
  Lemma~\ref{lem:loc} because $\vX$ and $\vY$ are independent and $\vY
  / \aPw{\vY}$ is rotationally symmetric. Equation~(\ref{proppow})
  implies that $\aPw{\vZ} \geq \aPw{\vY}$ since the function $ -\E{
    \ln p_{{\bf \atilde{Z}}} \left( \frac{\bf Z}{\text{P}} \right) }$
  is decreasing in $\text{P} \geq 0$.
 
\item[(iv)] Let $\vX$ and $\vY$ be two independent random vectors. If
  $\vY$ has a rotationally symmetric PDF that is non-increasing in $\|
  \vy \|$, then $\aPw{c \vX + \vY}$ is non-decreasing in $|c|$, $c \in
  \Reals$.

  We first show that $-\E{\ln p_{{\bf \atilde{Z}}}\left(\frac{c \vX +
        \vY }{\text{P}}\right)}$ is non-decreasing in $|c|$. To this
  end, we write
  \begin{equation*}
    -\E{\ln p_{{\bf \atilde{Z}}}\left(\frac{c \vX + \vY}{\text{P}}\right)} = \Ep{\vX}
    {-\Ep{\vY}{\ln p_{{\bf \atilde{Z}}} \left( \frac{c \vx + \vY}{\text{P}} \right) \Big| \vX }},
  \end{equation*}
  and it is enough to show that $- \Ep{\vY}{\ln p_{{\bf \atilde{Z}}}
    \left( \frac{c \vx + \vY}{\text{P}} \right)}$ is non-decreasing in
  $|c|$, which we argue as follows:
  \begin{equation}
    -\Ep{\vY}{\ln p_{{\bf \atilde{Z}}} \left( \frac{c \vx + \vY}{\text{P}}\right) } = 
    -\int_{\Reals^d} p_{\vY}(\vy) \ln p_{{\bf \atilde{Z}}}\left(\frac{c \vx + \vy}{\text{P}}\right)\, d \vy.
    \label{need}
  \end{equation}
  Since $p_{\vY}(\cdot)$ and $ p_{{\bf \atilde{Z}}}(\cdot)$ are
  rotationally invariant, one can restrict the proof to the case when
  $c \geq 0$ and the $\{x_i\}_{1 \leq i \leq d}$'s are non-negative by
  applying an appropriate rotation transformation to the variable of
  integration. Hence, for $c$ and $\{x_i\}_{1 \leq i \leq d}$
  non-negative, taking the derivative of equation~(\ref{need}) with
  respect to $c$ and interchanging the limit and the derivative as
  done in (ii) yields
  \begin{align*}
    & -\sum_{i = 1}^{d} \frac{x_i}{\text{P}} \int^{+\infty}_{-\infty} p_Y(y) \frac{\partial}{\partial z_i}p_{{\bf \atilde{Z}}}
    \left(\frac{c \vx + \vy}{\text{P}}\right) \frac{1}{p_{{\bf \atilde{Z}}}\left(\frac{c \vx + \vy}{\text{P}}\right)} \,dy \\
    =  & -\sum_{i = 1}^{d} \frac{x_i}{\text{P}} \, \E{ \frac{\partial}{\partial z_i}p_{{\bf \atilde{Z}}} \left( \frac{c \vx
          + \vy}{\text{P}} \right) \frac{1}{p_{{\bf \atilde{Z}}} \left( \frac{c \vx + \vy}{\text{P}}\right)}} \geq 0,
  \end{align*}
  which is true by virtue of the fact that $p_{\vY}(\vy)$ is
  rotationally symmetric, non-increasing in $\| \vy \|$, that for all
  $1 \leq i \leq d$ the derivative function $\frac{\partial}{\partial
    z_i}p_{{\bf \atilde{Z}}}\frac{1}{p_{{\bf \atilde{Z}}}} (\vx)$ is
  an odd function that is non-positive on $x_i \geq 0$ and that both
  $c$ and $\{x_i\}_{1 \leq i \leq d}$ are non-negative. This implies
  that both $-\Ep{\vY}{ \ln p_{{\bf \atilde{Z}}} \left( \frac{c \vx +
        \vY}{\text{P}} \right) }$ and $-\E{ \ln p_{{\bf \atilde{Z}}}
    \left( \frac{c \vX + \vY}{\text{P}} \right) }$ are non-decreasing
  in $|c|$. The fact that $-\E{ \ln p_{{\bf \atilde{Z}}} \left(
      \frac{c \vX + \vY}{\text{P}} \right) }$ is non-increasing in
  $\text{P}$ and non-decreasing in $|c|$ yields the required result.
  
\item[(v)] Whenever $\vX \sim \bm{\mathcal{S}} \left( \alpha,
    \gamma_{\vX} \right)$, $\aPw{\vX} = \frac{\gamma_{\bf
      X}}{\gamma_{\bf \atilde{Z}}} = (\alpha)^{\frac{1}{\alpha}}
  \gamma_{\vX}$. Indeed, $\vX$ has the same distribution as
  $\frac{\gamma_{\vX}}{\gamma_{\bf \atilde{Z}}} {\bf \atilde{Z}}$
  \begin{equation*}
    -\E{ \ln p_{\bf \atilde{Z}} \left( \frac{\vX}{\aPw{\vX}} \right) } =-\frac{\gamma_{\bf \atilde{Z}}}
    {\gamma_{\vX}} \int p_{\bf \atilde{Z}}\left( \frac{\gamma_{\bf \atilde{Z}}}{\gamma_{\vX}} \vx \right) \ln p_{\bf \atilde{Z}}
    \left(\frac{\vx}{\aPw{\vX}}\right)\, d \vx = h({\bf \atilde{Z}}),
  \end{equation*}
  and 
  therefore $\aPw{\vX} = \frac{\gamma_{\vX}}{\gamma_{\bf
      \atilde{Z}}}$.
\end{itemize}

\section{Sufficient Conditions for the regularity condition}
\label{appsuff}

In his technical report~\cite[sec. 6]{barrontech}, Barron proves that
the de Bruijn's identity for Gaussian perturbations~(\ref{debruijn}
with $\epsilon > 0$) holds for for any RV having a
finite variance. In this appendix, we follow steps similar to Barron's
to prove that condition~(\ref{condthrough}) is satisfied for any
$\left(\vX + \sqrt[\alpha]{\eta}\vN\right)$, $\eta > 0$ for any random
vector $\vX \in \mathcal{L}$ where
\begin{equation*}
  \mathcal{L} = \left\{ \text{random vectors} \,\,{\vU \in \Reals^{d}}: \int \ln\left(1 + \| \vU \|\right)\,
    dF_{\vU}(\vu) \text{ is finite } \right\},
\end{equation*} 
and where $\vN \sim \bm{\mathcal{S}} \left( \alpha, 1 \right)$ is
independent of $\vX$, $0< \alpha <2$.

In what follows, denote $q_{\eta}(\vy) = \E{ p_{\eta}(\vy-\vX) }$
be the PDF of $\vY = \vX + \sqrt[\alpha]{\eta}\vN$ where
$p_{\eta}(\cdot)$ is the density of the sub-Gaussian \SaS vector with
dispersion $\eta$. Note that since $p_{\eta}(\cdot)$ is bounded then
so is $q_{\eta}(\cdot)$ and since $\vX \in \mathcal{L}$ then so is
$\vY$. Then $h(\vY)$ is finite and is defined as
\begin{equation*}
  h(\vY) = -\int q_{\eta}(\vy) \ln q_{\eta}(\vy)\,d\vy.
\end{equation*}
We list and prove next two technical lemmas.

\begin{lemma}[Technical Result]
  \label{lemt1}
  \begin{equation*}
    \frac{d}{d \eta} q_{\eta}(\vy) = \E{ \frac{d}{d \eta} p_{\eta}(\vy-\vX) }.
  \end{equation*}
\end{lemma}

\begin{proof}
  The interchange between differentiation and integration holds
  whenever $|\frac{d}{d \eta} p_{\eta}(\vt)|$ is bounded uniformly
  by an integrable function in a neighbourhood of $\eta$ by virtue of
  the mean value theorem and the Lebesgue DCT.  To prove boundedness,
  we start by evaluating the derivative. Since
  \begin{equation*}
    p_{\eta} (\vt) = \frac{1}{\eta^{\frac{d}{\alpha}}}p_\vN\left(\frac{\vt}{\sqrt[\alpha]{\eta}}\right), 
  \end{equation*}
  then
  \begin{equation*}
    \frac{d}{d \eta} p_{\eta}(\vt) = -\frac{d}{\alpha}\frac{1}{\eta^{1+\frac{d}{\alpha}}}p_\vN\left(\frac{\vt}
      {\sqrt[\alpha]{\eta}}\right) -\frac{1}{\alpha}\frac{1}{\eta^{1+\frac{1+ d}{\alpha}}}\sum_{j =1}^{d} t_j 
    \left(\frac{\partial p_{\vN}}{\partial n_{j}}\right)_{\frac{\vt}{\sqrt[\alpha]{\eta}}},
  \end{equation*}
  which gives 
  \begin{equation}
    \left|\frac{d p_{\eta}}{d \eta} (\vt)\right| \leq \frac{d}{\alpha}\frac{1}{\eta^{1+\frac{d}{\alpha}}}p_\vN
    \left(\frac{\vt}{\sqrt[\alpha]{\eta}}\right) + \frac{1}{\alpha}\frac{1}{\eta^{1+\frac{1+d}{\alpha}}}\sum_{j =1}^{d} 
    \left|t_j\right| \left|\frac{\partial p_{\vN}}{\partial n_{j}}\right|_{\frac{\vt}{\sqrt[\alpha]{\eta}}}.  \label{deriv}
  \end{equation}
  For the purpose of finding the uniform bound on the derivative, let
  $b$ as a positive number such that $b < \eta$. Concerning the first
  term of the bound in~(\ref{deriv}), we consider two separate ranges
  of the variable $r = \| \vt \|$ to find the uniform upper
  bound. On compact sets, we have
  \begin{equation}
    \label{firstcom}
    \frac{d}{\alpha}\frac{1}{\eta^{1+\frac{d}{\alpha}}}p_\vN\left(\frac{\vt}{\sqrt[\alpha]{\eta}}\right) \leq \frac{d}{\alpha} 
    \frac{1}{b^{1+\frac{d}{\alpha}}}\max_{\vu \in \Reals^{d}} p_\vN(\vu),
  \end{equation}
  where the maximum exists since $p_{\vN}$ is a continuous PDF and
  thus upper bounded. As for large values of $\| \vt \|$, by virtue of
  equation~(\ref{eq:tailvec}) there exists some $k > 0$ such that
  $p_{\vN}(\vt) \leq k \frac{1}{\| \vt \|^{d+\alpha}}$ which gives
  \begin{equation}
    \label{firsth}
    \frac{d}{\alpha}\frac{1}{\eta^{1+\frac{d}{\alpha}}}p_\vN\left(\frac{\vt}{\sqrt[\alpha]{\eta}}\right) \leq k \frac{d}{\alpha}
    \frac{1}{\| \vt \|^{d +\alpha}},
  \end{equation}  
  an integrable upper bound function independent of
  $\eta$. Equations~(\ref{firstcom}) and~(\ref{firsth}) insures that
  the first term of the right-hand side (RHS) of
  equation~(\ref{deriv}) is uniformly upper bounded by an integrable
  function. When it comes to the second term of the RHS
  of~(\ref{deriv}), we formally have:
  \begin{equation*}
    \frac{\partial p_{\vN}}{\partial u_{j}}\left(\vu\right) = \frac{-i}{(2 \pi)^d} \int \omega _{j} \phi_{\vN}(\vomega) 
    e^{-i \sum _{l =1}^{d}\omega_{l} u_{l}} \,d\vomega, \quad \quad 1 \leq j \leq d
  \end{equation*}
  and
  \begin{equation}
    \label{uppcons}
    \left|\frac{\partial p_{\vN}}{\partial u_{j}}\left(\vu\right)\right| \leq \frac{1}{(2 \pi)^d} \int_{\Reals^d} 
    |\omega _{j}| e^{-\|\vomega\|^\alpha} \,d\vomega = \xi_j, \quad \quad \,\,\,\,\,\,\, 1 \leq j \leq d
  \end{equation}
  which is finite and where we used the fact that the characteristic
  function of $\bm {\mathcal{S}}(\alpha;1)$ is $\phi_{\vN}(\vomega) =
  e^{- \|\vomega\|^{\alpha}}$.  Hence, on compact sets,
  equation~(\ref{uppcons}) gives a uniform integrable upper bound on
  the second term of the RHS of equation~(\ref{deriv}) of the form
  \begin{equation}
    \label{secondcom}
    \frac{1}{\alpha}\frac{1}{\eta^{1+\frac{1+d}{\alpha}}}\sum_{j =1}^{d} \left|t_j\right| \left|\frac{\partial p_{\vN}}
      {\partial n_{j}}\right|_{\frac{\vt}{\sqrt[\alpha]{\eta}}} \leq \frac{1}{\alpha}\frac{1}{b^{1+\frac{1+d}{\alpha}}}\sum_{j = 1}^d |t_j| \xi_j,
  \end{equation}
  which is integrable and independent of $\eta$. Therefore, we only
  consider next the second term of the RHS of equation~(\ref{deriv})
  at large values of $\|t\|$. To this end, we make use of
  equation~(\ref{eq:lim}) proven in Appendix~\ref{asd} where it has
  been shown that $- \sum_{j =1}^{d} \frac{1}{t_j} \frac{ \partial
    p_{\vN}}{\partial t_j} (\vt) = \sum_{j =1}^{d} \frac{1}{|t_j|}
  \left| \frac{\partial p_{\vN}}{\partial t_j} \right|_{\vt} = \Theta
  \left( \frac{1}{\|t\|^{d + \alpha +2}} \right)$ and we write for
  some $\kappa > 0$
  \begin{equation}
    \label{secondh}
    \frac{1}{\alpha}\frac{1}{\eta^{1+\frac{1+d}{\alpha}}}\sum_{j =1}^{d} \left|t_j\right| \left|\frac{\partial p_{\vN}}{\partial n_{j}}
    \right|_{\frac{\vt}{\sqrt[\alpha]{\eta}}} \leq \frac{1}{\alpha}\frac{\| \vt \|^2}{\eta^{1+\frac{2+d}{\alpha}}} \sum_{j =1}^{d} 
    \frac{\sqrt[\alpha]{\eta}}{\left|t_j\right|} \left|\frac{\partial p_{\vN}}{\partial n_{j}}\right|_{\frac{\vt}{\sqrt[\alpha]{\eta}}}
    \leq \frac{1}{\alpha} \frac{\kappa}{\| \vt \|^{d+\alpha}},
  \end{equation}
  which is uniformly bounded at large values of $\|t\|$ by an
  integrable function. Equations~(\ref{secondcom}) and~(\ref{secondh})
  imply that the second term in the RHS of equation~(\ref{deriv}) is
  uniformly upper bounded by an integrable function. This proves
  Lemma~\ref{lemt1}.
\end{proof}

\begin{lemma}[Existence of the Generalized Fisher Information]
  \label{lemt2}

  The derivative
  \begin{equation*}
    \frac{d}{d \eta} h(\vX + \sqrt[\alpha]{\eta}\vN) =  -\int \frac{d}{d \eta} \left(q_{\eta}(\vy)\right)\,\ln q_{\eta}(\vy)\,d\vy
  \end{equation*}
  exists and is finite. Also,
  \begin{equation*}
    J_{\alpha}(\vX + \sqrt[\alpha]{\eta}\vN) = -\int \frac{d}{d \eta} \left(q_{\eta}(\vy)\right)\,\ln q_{\eta}(\vy)\,d\vy.
  \end{equation*}
\end{lemma}

\begin{proof}
  \begin{align}
    \frac{d}{d \eta} h(\vY) &= -\int \frac{d}{d \eta} \left(q_{\eta}(\vy)\,\ln q_{\eta}(\vy)\right)\,d\vy \label{inter11111}\\
    &= -\int \frac{d q_{\eta}}{d \eta}(\vy)\,\ln q_{\eta}(\vy)\,d\vy -\int \frac{d q_{\eta}}{d \eta}(\vy) \,d\vy \nonumber\\
    &= -\int \frac{d q_{\eta}}{d \eta}(\vy) \,\ln q_{\eta}(\vy)\,d\vy - \frac{d}{d \eta}\int q_{\eta}(\vy)\,d\vy \label{inter22222}\\
    &= -\int \frac{d q_{\eta}}{d \eta}(\vy) \,\ln q_{\eta}(\vy)\,d\vy. \label{final111}
  \end{align}
  Equation~(\ref{final111}) is true since $q_\eta(\vy)$ is a PDF and
  integrates to $1$. Next, we justify equation~(\ref{inter22222}):
  note that by Lemma~\ref{lemt1},
  \begin{eqnarray*}
    \left|\frac{d q_{\eta}}{d \eta} (\vy)\right| &=& \left| \E{ \frac{d p_{\eta}}{d \eta} (\vy-\vX) } \right|\\
    &\leq& \E{ \left| \frac{d p_{\eta}}{d \eta} (\vy-\vX) \right| },
  \end{eqnarray*}
  because the absolute value function is convex. Now it has been shown
  in the proof of Lemma~\ref{lemt1} that $\left|\frac{d p_{\eta}}{d
      \eta}(\vt)\right|$ is uniformly upper bounded in a
  neighbourhood of $\eta$ by an integrable function $s_b(\vt)$ of
  the form
  \begin{equation}
    \label{bydeff}
    s_b(\vt) = \left\{ \begin{array}{ll}
        \displaystyle  A(b) + B(b)\|\vt\| & \| \vt \| \leq r_0\\ 
        \displaystyle  \frac{C}{\| \vt \|^{d + \alpha}}& \| \vt \| \geq r_0,
      \end{array} \right.
  \end{equation} 
  where $A(b)$, $B(b)$, $C$ and $t_0$ are some positive values chosen
  in accordance with equations~(\ref{deriv}), (\ref{firstcom}),
  (\ref{firsth}), (\ref{secondcom}) and~(\ref{secondh}). Then
  \begin{equation*}
    \E{ \left| \frac{d p_{\eta}}{d \eta} (\vy-\vX)\right|} \leq \E{ \left| s_b(\vy-\vX) \right|} = r_b(\vy),
  \end{equation*}
  which is integrable by Fubini's theorem since $s_b(\vt)$ is
  bounded. This completes the justification of
  equation~(\ref{inter22222}).

  As for equation~(\ref{inter11111}), finding a uniform integrable
  upper bound to $\frac{d}{d \eta} \left(q_{\eta}(\vy) \, \ln
    q_{\eta}(\vy) \right)$ is achieved by finding one to $\frac{d
    q_{\eta}(\vy)}{d \eta} \, \ln q_{\eta}(\vy)$ which we show
  next. {Since $q_{\eta}(\vy)$ is continuous and positive, then it
    achieves a positive minimum on compact subsets of
    $\Reals^{d}$. Let $\vy_{0}$ be such that 
    \begin{align*}
      & \min_{\|\vy\| \leq \|\vy_{0}\|} q_{\eta}(\vy) \leq 1 \\
      \& \qquad & \max_{\|\vy\| \leq \|\vy_{0}\|} \left| \ln
        q_{\eta}(\vy) \right| \leq \left| \ln \min_{ \|\vy\| \leq
          \|\vy_{0}\|} q_{\eta}(\vy) \right|,
    \end{align*}
    then on $\|\vy\| \leq \|\vy_{0}\|$ we have
    \begin{align}
      \left|\frac{d  q_{\eta}(\vy)}{d \eta}\,\ln q_{\eta}(\vy)\right| &\leq \max_{\vy \in  \Reals^{d}} r_b(\vy) \left|
        \ln\min_{\|\vy\| \leq \|\vy_{0}\|}q_{\eta}(\vy)\right| \nonumber\\
      &\leq  \max_{\vy \in  \Reals^{d}} s_b(\vy) \left| \ln \min_{\|\vy\| \leq \|\vy_{0}\|}p_{\eta}(\vy)\right| \label{convrel}\\
      &\leq  \max_{\vy \in  \Reals^{d}} s_b(\vy)  \left|\ln\frac{1}{(2b)^{\frac{d}{\alpha}}} \, p_\vN\left(\frac{\vy_0}{\sqrt[\alpha]{b}}
        \right)\right| < \infty, \nonumber
    \end{align}
    which is independent of $\eta$.
    Equation~(\ref{convrel}) is justified by the fact that
    \begin{equation*}
      \min_{\|\vy\| \leq \|\vy_0\|} p_{\eta} (\vy) \leq \min_{\|\vy\|
        \leq \|\vy_0\|} q_{\eta} (\vy) \leq 1,
    \end{equation*}
    because $q_{\eta}(\vy) = \E{ p_{\eta}( \vy - \vX )}$. When it
    comes to large values of $\|\vy\|$, we have by the results of
    Property~\ref{prop:tail} in Appendix~\ref{asd} that $p_{\vN}({\bf
      t}) = \Theta \left( \frac{1}{\| \vt \|^{d+\alpha}} \right)$,
    and hence there exist positive $T$ and $K$ such that $p_{\vN}
    (\vt)$ is greater than $K \, \frac{1}{\| \vt \|^{d+\alpha}}$
    for some $K$ whenever $\| \vt \| \geq T$. Define $\tilde{\vy}$
    such that $\text{Pr}(\|\vX\| \leq \|\tilde{\vy}\|) \geq
    \frac{1}{2}$ and choose $\|\vy_0\|$ to be large enough. Then, if
    $b < \eta < 2b$, we have for $\|\vy\| \geq \|\vy_0\|$
    \begin{align*}
      q_{\eta}(\vy) = \, &  \frac{1}{\eta^{\frac{d}{\alpha}}} \int p_{\vN} \left(\frac{\vy- \vu}{\sqrt[\alpha]{\eta}}\right) 
      \,dF_\vX(\vu)\\
      \geq \, &  \frac{1}{\eta^{\frac{d}{\alpha}}} \int\limits_{\| \vu \| \leq \|\tilde{\vy}\|} p_{\vN} \left(\frac{\vy-\vu}
        {\sqrt[\alpha]{\eta}}\right) \, dF_\vX(\vu)  \\
      \geq \; &  \frac{1}{2 (2b)^{\frac{d}{\alpha}}} \, p_{N} \left(\frac{\|\vy\|+\|\tilde{\vy}\|}{\sqrt[\alpha]{b}}\right) \\
      \geq \; & \frac{b K}{2^{1+\frac{d}{\alpha}}} \frac{1} {\left(\|\vy\|+\|\tilde{\vy}\|\right)^{d + \alpha}} \\
      \geq \;&\frac{b \tilde{K}}{\|\vy\|^{d + \alpha}},
    \end{align*}
    where $\tilde{K}$ is some positive constant.
    At large values of $\|\vy\|$, $q_{\eta}(\vy) \leq 1$ and hence
    $|\ln q_{\eta}(\vy)| \leq \ln \left(\frac{\|\vy\|^{d + \alpha}}{b
        \tilde{K}}\right)$ and we obtain for $\|\vy\| > \|\vy_0\|$}
  \begin{equation*}
    \left|\frac{d  q_{\eta}(\vy)}{d \eta}\,\ln q_{\eta}(\vy)\right| \leq  r_b(\vy) \left(\ln \frac{\|\vy\|^{d + \alpha}}
      {b \tilde{K}}\right), 
  \end{equation*}
  which is a uniform integrable upper bound because
  \begin{align}
    \int \ln\left(1+\|\vy\|\right)r_b(\vy) \,d\vy 
    & = \iint \ln\left(1+\|\vy\|\right) s_b(\vy-\vx)\,dF_\vX(\vx)\,d\vy \nonumber\\
    & = \iint \ln\left(1+\|\vy\|\right) s_b(\vy-\vx)\,d\vy\,dF_\vX(\vx) \label{finter}\\
    & \leq \iint \left(\ln(1+\|\vx\|) + \ln(1+\|\vy\|)\right) s_b(\vy)\,d\vy\,dF_\vX(\vx)\nonumber\\
    & = S_b  \int \ln(1+\|\vx\|) dF_\vX(\vx) + L_b \nonumber\\
    & < \infty \label{ff}, 
  \end{align} 
  where 
  \begin{equation*}
    S_b = \int s_b(\vy)\,d\vy < \infty \qquad \& \qquad
    L_b = \int \ln(1+\|\vy\|)s_b(\vy)\,d\vy < \infty.
  \end{equation*}

  Note that $S_b$ and $L_b$ are finite
  by~(\ref{bydeff}). Equation~(\ref{finter}) is due to Fubini and
  equation~(\ref{ff}) is justified by the fact that $\vX \in
  \mathcal{L}$. In conclusion, equation~(\ref{inter11111}) is true and
  Lemma~\ref{lemt2} is proved.
\end{proof}

\bibliographystyle{IEEEtran}
\bibliography{paper}

\end{document}